\documentclass[11pt]{article}
\usepackage{mathtools}
\usepackage{color}
\usepackage{hyperref}
\hypersetup{colorlinks=true,linkcolor=black,citecolor=[rgb]{0.5,0,0},urlcolor=cyan}
\usepackage[framemethod=tikz]{mdframed}
\usepackage{wrapfig}

\usepackage[english]{babel}
\usepackage[utf8]{inputenc}
\usepackage{amsmath}
\usepackage{graphicx}
\usepackage{nicefrac}
\usepackage{amsmath,amsthm,amssymb,epsfig,enumerate,graphicx}
\usepackage{boxedminipage2e}
\usepackage{algpseudocode}
\usepackage{algorithm}
\usepackage{algorithmicx}
\usepackage{setspace}
\newtheorem{theorem}{Theorem}
\newtheorem{open}{Open Question}
\newtheorem{claim}[theorem]{Claim}
\newtheorem{lemma}[theorem]{Lemma}
\newtheorem{observation}[theorem]{Observation}
\newtheorem{proposition}[theorem]{Proposition}
\newtheorem{definition}[theorem]{Definition}
\newtheorem{corollary}[theorem]{Corollary}
\newtheorem{fact}[theorem]{Fact}
\newtheorem{remark}[theorem]{Remark}

\newcommand{\depth}{\mathsf{Depth}}

\newcommand{\RPC}{Replacement Path Covering}

\newcommand{\FT}{\mathsf{FT}}

\def\cA{{\cal A}}

\def\LAB{\mbox{\tt Label}}
\def\FTBFS{\mbox{\tt FT-BFS}}

\newcommand{\poly}{\mathsf{poly}}
\newcommand{\polylog}{\mathsf{polylog}}
\renewcommand{\paragraph}[1]{\vspace{0.15cm}\noindent {\bf #1}}

\def\Root{\mbox{\tt r}}
\def\Leaf{\mbox{\tt Leaf}}
\def\NLeaf{\mbox{\tt nLeaf}}
\def\NodesIn{\mbox{\tt N}}
\def\mod{\mbox{\tt mod}}

\def\SPT{\mbox{\tt SPT}}

\def\len{\textbf{w}}
\newcommand{\dist}{\mbox{\rm dist}}
\usepackage{xspace}
\def\distc{\Delta}
\newcommand{\abs}[1]{\left|#1\right|}
\newcommand{\sett}[2]{\left\{ #1 \left| \; \vphantom{#1 #2} \right. #2  \right\}} 
\newcommand{\rs}{\mathsf{C_{\textsf{RS}}}}
\newcommand{\con}{\mathsf{C_{\textsf{AG}\circ\textsf{RS}}}}
\newcommand\code[2]{$\left[ #1 \right]_{#2}$\xspace}
\renewcommand{\H}{\mathcal{H}}
\newcommand{\HM}{$\mathsf{HM}$ hash family\xspace}
\newcommand{\HMC}[1]{$\mathsf{HM}_2(#1)$\xspace}
\newcommand{\SHMC}[1]{$\mathsf{SHM}_2(#1)$\xspace}
\newcommand{\RP}{$\mathsf{RPC}$\xspace}
\newcommand{\CV}{$\mathsf{CV}$}
\newcommand{\DSO}{$\mathsf{DSO}$\xspace}

\usepackage{wrapfig}
\usepackage{geometry}
  \usepackage{nth}
  \usepackage{intcalc}

\begin{document}

\title{\textbf{Deterministic Replacement Path Covering}\footnote{An extended abstract of this article appeared in SODA'21.}}
\author{
Karthik C.\ S.\footnote{This work was   supported by the  Israel Science Foundation (grant number 552/16), the Len Blavatnik and the Blavatnik Family foundation and by   the Simons Foundation, Grant Number 825876, Awardee Thu D. Nguyen.} \\
        \small Tel Aviv University\\
        \small \texttt{karthik0112358@gmail.com} 
\and				
Merav Parter \footnote{Partially supported by the Israel Science Foundation (grant number 2084/18).}\\
        \small Weizmann Institute of Science \\
        \small \texttt{merav.parter@weizmann.ac.il}
}

\date{}
\maketitle

\begin{abstract}
In this article, we provide a unified and simplified approach to derandomize central results in the area of fault-tolerant graph algorithms. Given a graph $G$, a vertex pair $(s,t) \in V(G)\times V(G)$, and a set of edge faults $F \subseteq E(G)$, a replacement path $P(s,t,F)$ is an $s$-$t$ shortest path in $G \setminus F$. For integer parameters $L,f$, a \emph{replacement path covering} (\RP) is a collection of subgraphs of $G$, denoted by $\mathcal{G}_{L,f}=\{G_1,\ldots, G_r \}$, such that for every set $F$ of at most $f$ faults (i.e., $|F|\le f$) and every replacement path $P(s,t,F)$ of at most $L$ edges, there exists a subgraph $G_i\in \mathcal{G}_{L,f}$ that contains all the edges of $P$ and does not contain any of the edges of $F$. The covering value of the \RP $\mathcal{G}_{L,f}$ is then defined  to be the number of subgraphs in $\mathcal{G}_{L,f}$.\vspace{0.2cm}

\begin{sloppypar}In the randomized setting, it is easy to build an $(L,f)$-\RP with covering value of $O(\max\{L,f\}^{\min\{L,f\}}\cdot \min\{L,f\}\cdot \log n)$, but to this date, there is no efficient \emph{deterministic} algorithm with matching bounds. As noted recently by Alon, Chechik, and Cohen~(ICALP~2019) this poses the key barrier for derandomizing known constructions of distance sensitivity oracles and fault-tolerant spanners. We show the following:\end{sloppypar}
\begin{itemize}
\item There exist efficient deterministic constructions of $(L,f)$-\RP{}s whose covering values almost match the randomized ones, for a wide range of parameters. Our time and value bounds improve considerably over the previous construction of Parter (DISC 2019). Our algorithms are based on the introduction of a novel notion of hash families that we call \emph{Hit and Miss} hash families. We then show how to construct these hash families from (algebraic) error correcting codes such as Reed-Solomon codes and Algebraic-Geometric codes. 

\item For every $L,f$, and $n$, there exists an $n$-vertex graph $G$ whose $(L,f)$-\RP covering value is $\Omega(L^f)$. This lower bound is obtained by exploiting connections to the problem of designing 
sparse fault-tolerant BFS structures. 
\end{itemize}

An applications of our above deterministic constructions is the derandomization of the algebraic construction of the distance sensitivity oracle by Weimann and Yuster (FOCS 2010). The preprocessing and query time of  our deterministic algorithm nearly match the randomized bounds. This resolves the open problem of Alon, Chechik and Cohen (ICALP 2019).\vspace{0.2cm}

Additionally, we show a derandomization of the randomized construction of vertex fault-tolerant spanners by Dinitz and Krauthgamer (PODC 2011) and Braunschvig et al.\ (Theor.\ Comput.\ Sci., 2015). The time complexity and the size bounds of the output spanners nearly match the randomized counterparts.
\end{abstract}

\newpage
\tableofcontents
\newpage

\section{Introduction}
Resilience of combinatorial graph structures to faults is a major requirement in the design of modern graph algorithms and data structures. The area of fault tolerant (FT) graph algorithms is a rapidly growing subarea of network design in which resilience against faults is taken into consideration. The common challenge
addressed in those algorithms is to gain immunity against all possible fault events without losing out on the efficiency of the computation. Specifically, for a given graph $G$ and some bound $f$ on the number of faults, the FT-algorithm is required, in principle, to address all ${|E(G)|}\choose{f}$ fault events, but (usually) using considerably less space and time. The traditional approach to 
mitigate these challenges is based on a combinatorial exploration of the structure of the graph under faults. While this approach has led to many exciting results in the area, it is however limited in two aspects. First, in many cases the combinatorial characterization is considerably harder when moving from a single failure event to events with two or more failures. Second, this characterization is mostly problem specific and rarely generalizes to more than one class of problems.

One of the most notable techniques in this area which overcomes the aforementioned two limitations is  the fault-tolerant sampling technique introduced by Weimann and Yuster \cite{weimann2013replacement}. This technique is inspired by the color-coding technique \cite{alon1995color}, and provides a  general recipe for translating a given fault-free algorithm for a given task into a fault-tolerant one while paying a relatively small overhead in terms of computation time and other complexity measures of interest (e.g., space). Indeed this approach has been applied in the context of distance sensitivity oracles \cite{GrandoniW20J,GrandoniW20J,ChechikC20}, fault-tolerant spanners \cite{dinitz2011fault,BraunschvigCPS15,dinitz2020efficient}, fault-tolerant reachability preservers \cite{ChakrabortyC20}, distributed minimum-cut computation \cite{parter2019small}, and resilient distributed computation \cite{ParterYPODC19,ParterYSODA19,ChuzhoyPT20,BCHPArxiv20}.  The high-level idea of this technique is based on sampling a (relatively) small number of subgraphs $G_1,\ldots, G_\ell$ of the input graph $G$ by oversampling edges (or nodes) to act as faulty-edges, in a way that a single sampled subgraph accounts for potentially many fault events. An additional benefit of this approach is that it smoothly extends to accommodate multiple edge and vertex faults.

Two central applications of the above approach that we focus on are  distance sensitivity oracles and fault-tolerant spanners. An $f$-sensitivity distance oracle ($f$-\DSO) is a data-structure that reports shortest path distances when at most $f$ edges of the graph fail. Weimann and Yuster \cite{weimann2013replacement} employed the above technique to provide the first randomized construction of $f$-\DSO for $n$-vertex directed graphs accomodating $f=O(\log n/\log\log n)$ many number of faults. Their data-structure has subcubic preprocessing time and subquadratic query time, and these bounds are still the state-of-the-art results for a wide range of parameters. Recently, van-den Brand and Saranurak \cite{van2019sensitive} presented a randomized monte-Carlo \DSO that can handle $f\geq \log n$ updates. For small edge weights, their bounds improve over \cite{weimann2013replacement}. 
For the single failure case, Grandoni and Williams \cite{GrandoniW20J} also employed the sampling technique to provide an improved $1$-\DSO with subquadratic preprocessing time and sublinear query time. Very recently, Chechik and Cohen \cite{ChechikC20} improved their construction and obtained subcubic preprocessing time with $\widetilde{O}(1)$ query time. 
Since the key randomized component in these \DSO constructions is  the sampling of the subgraphs $\{G_i\}_{i\in[\ell]}$, Alon, Chechik and Cohen \cite{AlonCC19} posed the following question (stated specifically here for $f$-\DSO{}s):

\begin{center}
\emph{``It remains an open question if there exists a \DSO with subcubic deterministic\\ preprocessing algorithm and subquadratic deterministic query algorithm,\\ matching their randomized equivalents''.}
\end{center}

Another important application of this sampling technique appears in the context of fault-tolerant spanners. 
Given an $n$-vertex graph $G$, and integer parameters $f$ and $k$, an $f$-fault-tolerant $k$-spanner $H \subseteq G$ is a subgraph that contains a $k$-spanner in $G \setminus F$ for any set $F \subseteq V$ of at most $f$ vertices in $G$. The problem of designing sparse fault-tolerant spanners resilient to vertex faults was introduced by Chechik et al. \cite{chechik2010fault}. Using a careful combinatorial construction they showed that one can build such spanners while paying an additional overhead of $k^f$ in the size of the output spanner (when compared to the standard $k$-spanner). Dinitz and Krauthgamer \cite{dinitz2011fault} simplified and improved their construction. Using the sampling technique with the right setting of parameters, they provided a meta-algorithm for constructing fault-tolerant spanners where the time and size overheads are bounded by the factor $O(k^{2-1/f})$. Their approach was later  extended by Braunschvig et al. \cite{BraunschvigCPS15} to provide the first (and currently state-of-the-art) constructions of nearly-additive fault-tolerant spanners.  Very recently, Chakraborty and Choudhary \cite{ChakrabortyC20} employed this technique to provide a randomized construction of strong-connectivity
preservers of directed graphs under $f$ failures with $\widetilde{O}(f 2^f \cdot n^{2-1/f})$ edges.
To this date, there are no known efficient deterministic constructions that match the size bounds of these above-mentioned randomized constructions.

In this work we provide a unified and simplified approach for derandomizing the above mentioned central results. We introduce the notion of \emph{replacement path covering} (\RP{}) which captures the key properties of the collection of sampled subgraphs obtained by the FT-sampling technique. Given a graph $G$, a vertex pair $(s,t) \in V(G)\times V(G)$, and a set of edge faults $F \subseteq E(G)$, a replacement path $P(s,t,F)$ is an $s$-$t$ shortest path in $G \setminus F$.  To avoid repetitive descriptions, we mostly consider in this paper the setting of edge faults. However, all our definitions of \RP and their constructions naturally extend to vertex faults.
\begin{definition}[Replacement Path Covering (\RP)]\label{def:RP}
A subgraph $G' \subseteq G$ \emph{covers} a replacement path $P(s,t,F)$ if 
$P(s,t,F) \subseteq G' \mbox{~~and~~} F \cap E(G') =\emptyset.$

A collection of subgraphs of $G$, say $\mathcal{G}_{L,f}$, is an $(L,f)$-\emph{\RP} if for every $s,t \in V$ and every $F\subseteq E$ such that $|F|\leq f$, we have that each $P(s,t,F)$ replacement path\footnote{In case there are multiple $s$-$t$ shortest paths in $G \setminus F$ with at most $L$ edges, it is sufficient to cover one of them.} with at most $L$ edges is covered by some subgraph $G'$ in $\mathcal{G}_{L,f}$. The \emph{covering value} (\CV) of an $(L,f)$-\RP $\mathcal{G}_{L,f}$ is the number of subgraphs in $\mathcal{G}_{L,f}$, i.e., \CV($\mathcal{G}_{L,f}$):=$|\mathcal{G}_{L,f}|$.
\end{definition}

In some algorithmic applications of $(L,f)$-\RP, we have that $L \leq f$ and in others applications we have $L >f$. However, for simplicity of the discussion of this paragraph, we assume that $L >f$. The FT-sampling technique provides an efficient randomized procedure for computing 
an $(L,f)$-\RP of covering value $r=c\cdot fL^f\log n$ for some constant $c$ (e.g., Lemma 2 in \cite{GrandoniW20J}): Sample $r$ subgraphs $G_1,\ldots, G_r$ where each $G_i \subseteq G$ is formed by sampling each edge $e \in E(G)$ into $G_i$ independently with probability
$p=1-1/L$. By taking $c$ to be large enough, it is easy to show that a subgraph $G_i$ covers a fixed $P(s,t,F)$ with probability of $\Omega(1/L^f)$. Thus by using Chernoff and employing the union bound over all $n^{O(f)}$ distinct
$P(s,t,F)$ paths, one gets that
this graph collection is an $(L,f)$-\RP, with high probability (see Lemma \ref{lem:rand-RPC} for a formal proof). The computation time of this randomized procedure is $O(r \cdot m)$ (where $m:=|E(G)|$).  Alon, Chechik and Cohen \cite{AlonCC19} noted that in many settings, the deterministic computation of $(L,f)$-\RP poses the main barrier for derandomization, and raised the following question:\\

\begin{center}
\emph{``What is the minimum $r$ such that we can deterministically compute\\ such graphs $G_1\ldots, G_r$ in $\widetilde{O}(n^2 r)$ time such that for every $P(s,t,F)$\\ on at most $L$ nodes there is a subgraph $G_i$ that does not contain $F$\\ but contains $P(s,t,F)$?''}
\end{center}

\cite{AlonCC19} also mentioned that it is not clear \emph{how to efficiently derandomize a degenerated version of the above construction} and proposed some relaxation of these requirements, for which we indeed obtain improved bounds in this paper. 

Independently to the work of \cite{AlonCC19}, Parter \cite{parter2019small} recently provided\footnote{In \cite{parter2019small}, the term $(L,f)$-\RP is not used, and instead the deterministic algorithm is referred to as a derandomization of the FT-sampling technique.} a deterministic construction of $(L,f)$-\RP for the purposes of providing an efficient \emph{distributed} computation of small cuts in a graph. These \RP{}s are obtained by introducing the notion of $(n,k)$ universal hash functions. For the purpose of small cuts computation, $L$ was taken to be the diameter of the graph, and $f$ was considered to be constant. The goal in 
\cite{parter2019small} was to provide an $(L,f)$-\RP of value $\poly(L)$. Their construction in fact yields a value of $L^{4f+1}$. This value is already too large for several applications such as the \DSO by \cite{weimann2013replacement}. Indeed, for our \emph{centralized} applications, it is desirable to improve both the computation time as well as the covering value of these $(L,f)$-\RP constructions, and to match  (to the extent possible) the bounds of their randomized counterparts.

%
\subsection{Our Contributions}
We take a principled approach for efficiently computing almost optimal $(L,f)$-\RP for a wide range of parameters of interest.  Our algorithms extend the approach of \cite{parter2019small} and are based on the introduction of a novel notion of hash families that we call \emph{Hit and Miss} ($\mathsf{HM}$) hash families. We show how any Boolean alphabet \HM can be used to build a \RP, and in turn give near optimal constructions of  \HM based on (algebraic) error correcting codes such as Reed-Solomon codes and Algebraic-Geometric codes.  Our key result is as follows: 

\begin{theorem}[$(L,f)$--\RP]\label{thm:general-covering}
Given a graph $G$ on $m$ edges, length parameter $L$, and fault parameter $f$, there is a deterministic algorithm $\mathcal{A}$ for computing an     $(L,f)$-\RP of $G$ denoted by $\mathcal{G}_{L,f}$ such that,
  \[
    \text{\CV}(\mathcal{G}_{L,f})\le \left\{\begin{array}{lr}
       (\alpha c Lf)^{b+1}, & \text{if }a\ge m^{\nicefrac{1}{c}},\ \text{for some constant }c\in\mathbb{N},\\
        (\alpha Lf)^{b+2}\cdot \log m, & \text{if } a=m^{o(1)}\text{ and }b=\Omega(\log m),\\
                (\alpha Lf)^{b+2}\cdot \log m, & \text{if } a\le \log m,\\
       (\alpha Lf\log m)^{b+1}, & \text{otherwise,}
        \end{array}\right.
  \]
where $a=\max\{L,f\}$, $b=\min\{L,f\}$,  and $\alpha\in\mathbb N$ is some small universal constant.  
Moreover, the running time of $\mathcal{A}$ denoted by $T(\mathcal{A})$ is,
$$
T(\mathcal{A})=
\left\{\begin{array}{lr}
 m^{1+o(1)}\cdot \text{\CV}(\mathcal{G}_{L,f})&\text{if } a=m^{o(1)}\text{ and }b=\Omega(\log m),\\
m\cdot (\log m)^{O(1)}\cdot \text{\CV}(\mathcal{G}_{L,f}), & \text{otherwise.}
        \end{array}\right.
$$
\end{theorem}

This resolves the open problem of Alon, Chechik and Cohen \cite{AlonCC19} and considerably improves over the bounds of the second author \cite{parter2019small} in the entire range of parameters.  We further improve on the parameters of Theorem~\ref{thm:general-covering} (see Theorem~\ref{thm:relaxed-RPC}) when instead of accounting for all fault events, we only have to be resilient to a list of fault events that are given to us. Even this relaxed version was mentioned   in \cite{AlonCC19}.
%

At a meta level, \RP{}s are designed to handle faults in graphs, and error correcting codes are constructed to handle errors in messages. Both do this by adding redundancy to the underlying information in some way: the encoding of a message adds many new coordinates to the message without adding any new additional information, and similarly \RP{} of a graph is a redundant way to represent a graph, as we only store subgraphs of the same original graph. In this work, we formalize this meta-connection   to an extent through the ideas involved in proving Theorem~\ref{thm:general-covering}.

\paragraph{Lower Bound for $(L,f)$-\RP{}s.} We also prove lower bounds on the covering value of \RP, which to the best of our knowledge had not been addressed before. That is, despite the ubiquity of the FT-sampling approach to build $(L,f)$-\RP{}s, it is still unclear whether the bound that it provides on the covering value is the best possible. This question is interesting even if the items to be covered correspond to arbitrary subsets of edges. The question becomes even more acute in our setting where the covered items are structured, i.e., correspond to shortest-paths in some underlying subgraphs.
The optimality of the randomized procedure in this context is even more questionable, as it is totally invariant to the structure of the graph. In principle, one might hope to improve these bounds by taking the graph structure into account. 

Perhaps surprisingly we show that the covering values obtained by the randomized FT-sampling procedure are nearly optimal, at least for the setting where $L \geq f$. Since our deterministic bounds almost match the randomized ones, we obtain almost-optimality for our bounds.
\begin{theorem}[Lower Bound for the Covering Value of $(L,f)$-\RP]\label{thm:lower-bound}
For every integer parameters $n$, $L,$ and $ f$ such that $(L/f)^{f+1}\leq n$, there exists an $n$-vertex weighted graph $G^*=(V,E,w)$, such that any $(L,f)$-\RP of $G$ has \CV\ of $\Omega((L/f)^f)$.
\end{theorem}
Interestingly, the lower bound graph is obtained by employing slight modifications to the lower bound graphs used by \cite{parter2015dual} in the context of fault-tolerant \emph{FT-BFS} structures. 
For a given (possibly weighted) graph $G=(V,E)$ and a source vertex $s \in S$, a subgraph $H \subseteq G$ is an $f$-fault-tolerant (FT)-BFS if $\dist(s,t,H \setminus F)=\dist(s,t,G \setminus F)$ for every vertex $t \in V$ and every sequence of $F$ edge faults. The definition can be naturally extended to vertex faults as well. The second author and Peleg \cite{ParterP16} presented a lower-bound construction for $f=1$ with $\Omega(n^{3/2})$ edges. The second author extended this lower bound construction to any $f\geq 1$ faults with size bounds of $\Omega(n^{2-1/(f+1)})$ edges \cite{parter2015dual}. We show that a slight modification to the (unweighted) lower-bound graph of \cite{parter2015dual} by means of introducing weights, naturally implies a lower bound for the covering value of an $(L,f)$-\RP.

\paragraph{Derandomization of the Algebraic \DSO by Weimann-Yuster.}
Our key  application of the construction of efficient $(L,f)$-\RP is for implementing the \emph{algebraic} \DSO of \cite{weimann2013replacement}. \cite{AlonCC19} presented a derandomization of the combinatorial $f$-\DSO of \cite{weimann2013replacement}, resulting with a preprocessing time of $\widetilde{O}(n^{4-\alpha})$ and a query time of $\widetilde{O}(n^{2-2\alpha/f})$, matching the randomized bounds of \cite{weimann2013replacement}. 
In this paper we focus on derandomizing the algebraic algorithm of \cite{weimann2013replacement} as the latter can be implemented in subcubic preprocessing time and subquadratic query time. We show:
\begin{theorem}\label{thm:alg-dso-wy}\begin{sloppypar}
Let $G = (V, E)$ be a directed $n$-vertex $m$-edge graph with real edge weights in $[-M,M]$. There exists a deterministic algorithm that given $G$ and parameters $f = O(\log n/\log\log n)$  
and $0 < \alpha < 1$, constructs an $f$-sensitivity distance oracle in time 
\begin{enumerate} 
\item $O(M n^{3.373+2/f-\alpha}\cdot (c'f)^{f+1})$ if $\alpha=1/c$ for some constant $c$,
\item  $O(M n^{3.373+2/f-\alpha}\cdot (c'f \log n)^{f+1})$ if $\alpha=o(1)$,
\end{enumerate}
for some constant $c'$. Given a query $(s, t, F)$ with $s, t \in V$ and $F \subseteq E \cup V$  being a set of at most $f$ edges or vertices, the deterministic query algorithm computes in $O(n^{2-2(1-\alpha)/f})$ time the distance from $s$ to $t$ in the graph $G \setminus F$. \end{sloppypar}
\end{theorem}
Observe that for constant number of at least $f\geq 7$ faults, the preprocessing time of our construction even improves over that of Weimann-Yuster when fixing the query time to be $O(n^{2-2(1-\alpha)/f})$. This is because our algorithm also integrates ideas and optimizations from \cite{AlonCC19} and \cite{ChechikC20}.
This resolves the open problem of \cite{AlonCC19} concerning existence of deterministic \DSO with subquadratic preprocessing time and subquadratic query time (at least with small edge weights). 

While the deterministic $(L,f)$-\RP of Theorem \ref{thm:general-covering} constitutes the key tool for the derandomization, the final algorithm requires additional effort. Specifically, we use the notion of FT-trees introduced in \cite{AlonCC19} for the purpose of the deterministic combinatorial \DSO. We provide an improved algebraic construction of these trees using the $(L,f)$-\RP{}s. One obstacle that we need to handle is that the approach of \cite{AlonCC19} assumed that shortest path are unique by providing an algorithm that breaks the ties in a consistent manner. In our setting, the computation time of this algorithm is too heavy and thus we avoid this assumption, by making more delicate arguments. 

\paragraph{Derandomization of Fault-Tolerant Spanner Constructions.} Finally, we show that the integration of the $(L,f)$-\RP of Theorem \ref{thm:general-covering} into the existing algorithms for (vertex) fault-tolerant spanners provide the first deterministic constructions of these structures. The running time and the size bounds of the spanners nearly match the one obtained by the randomized counter parts. 
Specifically, for $f$-fault tolerant multiplicative spanners, we provide a nearly-optimal derandomization of the Dinitz and Krauthgamer's construction \cite{dinitz2011fault}. This follows directly by using our vertex variant of $(L,f)$-\RP of Theorem \ref{thm:general-covering} with $L=2$. A subgraph $H \subseteq G$ is an $f$-fault tolerant $t$-spanner if $\dist(s,t,H \setminus F)\leq t \cdot\dist(s,t,G\setminus F)$ for every $s,t, F \subseteq V$, $|F|\leq f$. We show:
\begin{theorem}[Derandomized of Theorem 2.1 of \cite{dinitz2011fault}, Informal]\label{thm:spanner-mult}
If there is a deterministic algorithm $\cA$ that on every $n$-vertex $m$-edge graph builds a $t$-spanner of size $s(n)$ and time $\tau(n,m,t)$, then there is an algorithm that on any such graph builds an $f$-fault tolerant $t$-spanner of size $\widetilde{O}(f^3 \cdot s(2n/f))$ and time $\widetilde{O}(f^3 (\tau(2n/f,m,t)+ m))$.
\end{theorem}
The above derandomization matches the randomized construction of \cite{dinitz2011fault} upto a multiplicative factor  of $\log^3 n$ in the size and time bounds. 
In the same manner, we also apply derandomization for the nearly-additive fault-tolerant spanners of Braunschvig et al. \cite{BraunschvigCPS15}. This provides the first deterministic constructions of nearly additive spanners.   

\paragraph{Comparison with a recent independent work of \cite{BodwinDinitz20}.}
Independent to our work, \cite{BodwinDinitz20} presented a new slack version of the greedy algorithm from \cite{bodwin2018optimal,DinitzR20} to obtain a (vertex) fault-tolerant spanners with \emph{optimal} size bounds. Their main algorithm is randomized with and the emphasis there is on optimizing the size of the output spanner. 
To derandomize their construction, \cite{BodwinDinitz20} used the notion of universal hash functions to compute  deterministically an $(L=2,f)$-\RP of covering value $\widetilde{O}(f^6)$ for $f\leq n^{o(1)}$ and a value of
$\widetilde{O}(f^3)$ for $f\geq n^c$ for some constant $c$. Using our $(L=2,f)$-RPC of Theorem \ref{thm:general-covering} yields a covering value of $\widetilde{O}(f^3)$ for \emph{every value} $f$. Up to a logarithmic factor, our bounds match the value of the randomized construction.  
The quality of the spanner construction of \cite{BodwinDinitz20} depends, however, not only on the value of the covering, but rather also on additional useful properties. These properties are also addressed in our paper for the sake of the applications of derandomizing the works of \cite{dinitz2011fault,weimann2013replacement}.
In particular, we show that our $(L=2,f)$-RPC with $\widetilde{O}(f^3)$ subgraphs also satisfies the desired properties in the same manner as provided by the randomized construction. Consequently, by using our $(L=2,f)$-RPCs in the algorithm of \cite{BodwinDinitz20},  we can close the gap of Theorem 1.2 of \cite{BodwinDinitz20} and get a deterministic construction which matches the randomized time bounds (of Theorem 1.1 in \cite{BodwinDinitz20}) for \emph{any} value of $f$. In Appendix~\ref{sec:comparison}, we provide a further detailed comparison to the related constructions of \cite{parter2019small} and \cite{BodwinDinitz20}. In addition, we provide a proof sketch for improving Thm. 1.2 of \cite{BodwinDinitz20} (see Lemma \ref{lem:imp}).  We also point the reader to subsequent work by Parter \cite{P22}.

\subsection{Key Techniques}

In this section, we detail some of the key techniques introduced in this paper.

\subsubsection{Deterministic $(L,f)$-\RPC}

While the introduction of the notion of \RP is our key conceptual contribution, we elaborate in this subsection on our framework to construct deterministic \RP, which we also believe will be of independent interest. 

\paragraph{Hit and Miss Hash Families.} We introduce a new notion of hash families called Hit and Miss ($\mathsf{HM}$) Hash Families. Informally, given integer parameters $N,a,b,$ and $q$, a family $\H$ of hash functions from $[N]$ to $[q]$ is said to be a \HM if for every pair of mutually disjoint subsets of $[N]$, say $(A,B)$, there exists a hash function  $h\in \H$ such that every $(x,y)\in A\times B$ do not collide under $h$ (see Definition~\ref{def:HM} for a formal statement). We show that every error correcting code with relative distance greater than $1-\frac{1}{ab}$ can be seen as a \HM. This insight yields a systematic way to construct \HM.

\paragraph{Connection to \RPC.} We then consider \HM over the Boolean alphabet and associate the domain of the hash family with the edges (or vertices) of the graph for which we would like to design a \RP. We observe that every hash function of the Boolean \HM immediately gives a subgraph in \RP, where we view the function as a Boolean vector of length equal to the number of edges in the graph, and thus the hash function acts as an indicator vector of whether to pick the edge or not in the subgraph. Moreover, the property of a \RP always avoiding faults but containing the replacement path in at least one of the subgraphs (see Definition~\ref{def:RP}) exactly coincides with the definition of a Boolean \HM, and thus a Boolean \HM yields a \RP.

\paragraph{Overview.} We now provide a short summary of our deterministic construction of ($L,f$)-\RP (assuming $L\ge f$) for a graph $G$ with $m$ edges.  We start from an error correcting code $C$ over alphabet of size $q$, block length $\ell$, message length $\log_q m$ and relative distance greater than $1-\frac{1}{Lf}$. Next, we  interpret $C$ as a \HM from $[m]$ to $[q]$ with $\ell$ hash functions. Then we apply the alphabet reduction lemma to obtain a  \HM from $[m]$ to $\{0,1\}$ with $\ell\cdot q^f$ many hash functions. Finally, using the connection between Boolean \HM and \RPC, we construct an $(L,f)$-\RP $\mathcal{G}_{L,f}$ with covering value $2\cdot q^f\cdot \ell$ in time $\mathsf{CV}(\mathcal{G}_{L,f})\cdot \widetilde{O}(m)$. In other words the alphabet size and the block length of the starting code $C$ directly determines the covering number of our \RP. Depending on the relationship between $L$ and $f$ we use either just Reed-Solomon code or a concatenation of Algebraic-Geometric code (as outer code) with   Reed-Solomon code (as inner code) to obtain the parameters given in Theorem~\ref{thm:general-covering}.

\subsubsection{Derandomization of Weimann-Yuster \DSO}
Our key contribution is in utilizing the $(L,f)$-\RP to compute \emph{fault-tolerant trees} with improved time bounds compared to that of \cite{AlonCC19}. Fault tolerant trees were introduced by \cite{ChechikCFK17,AlonCC19} and specifically, in \cite{AlonCC19} they served the basis for implementing the combinatorial \DSO implementation of \cite{weimann2013replacement}. For a given vertex pair $s,t$, and integer parameters $L,f$, the FT-tree $\FT_{L,f}(s,t)$ consists of $O(L^f)$ \emph{nodes}, where each node is labeled by a pair $\langle P,F \rangle$ where $P$ is an $s$-$t$ path in $G \setminus F$ with at most $L$ edges, where $F$ is a sequence of at most $f$ faults which $P$ avoids. 
Let $d^{L}(s,t,G')$ denote the weight of the shortest $s$-$t$ paths in $G'$ among all $s$-$t$ paths with at most $L$ edges. The key application of FT-trees is that given a query $(s,t,F)$ and the FT-tree $\FT_{L,f}(s,t)$, one can compute $d^{L}(s,t,G \setminus F)$ in time $O(f^2\log n)$. 
\cite{AlonCC19} provided an efficient combinatorial construction of all the FT-trees in time $\widetilde{O}(m\cdot n \cdot L^{f+1})$, thus super-cubic time for dense graphs.

By using our $(L,f)$-\RP family $\mathcal{G}_{L,f}$, we provide an improved (algebraic) construction of these trees in sub-cubic time for graphs with small integer weights. The construction of these trees boils down into a simple computational task which we can efficiently solve using the $(L,f)$-\RP. The task is as follows: given a triplet $(s,t,F)$, compute $d^L(s,t,G\setminus F)$. To build the trees, it is required to solve this task for $O(n^2 \cdot L^f)$ triplets. Our algorithm starts by applying a variant of the All-Pair-Shortest-Path (APSP) in each of the subgraph $G' \in \mathcal{G}_{L,f}$. This variant, noted as $APSP^{\leq L}$ \cite{ChechikC20} restricts attention to computing only the shortest paths that contain at most $L$ edges, which can be done in time $\widetilde{O}(M L n^{\omega})$ using matrix multiplications. 
Then to compute $d^L(s,t,G\setminus F)$ for a given triplet $(s,t,F)$, we show that it is sufficient to consider a small collection of subgraphs $\mathcal{G}_F \subseteq \mathcal{G}_{L,f}$ where $|\mathcal{G}_F|=O(f L\log n)$, and to return the minimum $d^L(s,t,G'')$ over every $G'' \in \mathcal{G}_F$. Since the $d^L(s,t,G')$ distances are precomputed by the  $APSP^{\leq L}$ algorithm, each $d^L(s,t,G\setminus F)$ can be computed in $\widetilde{O}(L)$ time.


\subsection{Gap between Det. and Randomized $(L,f)$-\RPC}\label{sec:random}
For the sake of discussion assume that $f=O(1)$ and $L=n^{\epsilon}$ for some constant $\epsilon$.  
Our current deterministic constructions provide $(L,f)$-\RP with covering value $\widetilde{O}(L^{f+1})$ whereas the randomized constructions obtain value of $\widetilde{O}(L^{f})$. This gap is rooted in the following distinction between the randomized and deterministic constructions. For the purposes of  the randomized construction, the $(L,f)$-\RP should cover $n^{O(f)}$ replacement paths. The reason is that there are $n^{O(f)}$ possible fault events, and for each sequence of $F$ faults, the subgraph $G \setminus F$ contains $n^2$ shortest paths (i.e., replacement paths avoiding $F$). In particular, if there are multiple $s$-$t$ shortest-path in $G \setminus F$, it is sufficient for the \RP to cover one of them. Since a single sampled subgraph $G_i$ covers a given path $P(s,t,F)$ with probability of $c/L$, by taking $r=O(fL^f \log n)$ subgraphs, we get that $P(s,t,F)$ is covered by at least one of the subgraphs with probability of $1-1/n^{c\cdot f}$. Applying the union bound over all $n^{O(f)}$ replacement paths establishes the correctness of the construction. In contrast, our deterministic construction provides a covering for any $P(s,t,F)$ paths, and also for any arbitrary collection of $L$ edges $A$ and $f$ edges $B$ with $A \cap B =\emptyset$. That is, since our construction does not exploit the structure of the paths, it provides a covering for $n^{\Omega(L)}$ paths. Note that if the randomized construction would have required to cover $n^{\Omega(L)}$ paths rather than $n^{O(f)}$, we would have end-up having $O(L^{f+1})$ subgraphs in that covering as well. 
In other words, the current gap in the bounds can be explained by the number of replacement paths that the $(L,f)$-\RP are required to cover. Since in the deterministic constructions, it is a-priori unknown what would be the set of replacement paths that are required to be covered, they cover all $n^{\Omega(L)}$ possible paths.

Importantly, in Appendix~\ref{sec:input-RPC}, we consider a relaxed variant of the $(L,f)$-RPC problem, introduced by \cite{AlonCC19}, for which we are able to provide nearly matching bounds to the randomized construction. Specifically, in that setting, we are given as input a collection of $k$ pairs $\{(P, F )\}$ where $P$ is a path with at most $L$ edges and $F$ is a set of at most $f$ faults which $P$ avoids. We then provide an efficient deterministic construction of a restricted $(L,f)$-\RP family $\mathcal{G}$ of value $\widetilde{O}(\log k \cdot L^f)$, i.e., of the same value as obtained by the randomized construction. The graph collection $\mathcal{G}$ then satisfies that for every pair $( P, F )$ in the input set, there is a subgraph $G' \in \mathcal{G}$ such that $P \subseteq G'$ and $G'\cap F=\emptyset$. This further demonstrates that the only reason for the gap between our deterministic and randomized bounds is rooted in the gap in the number of replacement paths that those constructions are required to cover.

\section{Preliminaries}

\paragraph{Notations.} Throughout this paper, $G$ denotes a (possibly weighted) graph, $V(G)$ denotes the vertex set of a graph $G$, and  $E(G)$ denotes the edge set of a graph $G$.  In case the graph is weighted, the weights are integers in $[-M,M]$. 
For $u,v \in V$ and a subgraph $G'$, let $\dist(u,v,G')$ denote the shortest $u$-$v$ path distance in $G$. For an $x$-$y$ path $P$ and $y$-$w$ path $P'$, let $P\circ P'$ denote the concatenation of the two paths. Also, for any $n\in\mathbb{N}$ and $j\in\mathbb{N}$, we denote by $\binom{[n]}{j}$ the collection of all subsets of size exactly $j$, by $\binom{[n]}{\le j}$ the collection of all subsets of size at most $j$,  and by $\binom{n}{\le j}$ the sum $\sum_{i\in [j]}\binom{n}{i}$.

\subsection{Replacement Paths and Randomized $(L,f)$ Covering}
For a weighted graph $G=(V,E,w)$ and a path $P \subseteq G$, let $|P|$ be the number of edges in $P$ and let $\len(P)=\sum_{e \in P}w(e)$ be the weighted sum of the edges in $P$. 
Let $SP_G(s,t,F)$ be the collection of all $s$-$t$ shortest path in $G \setminus F$. Every path $P_G(s,t,F)\in SP_G(s,t,F)$ is called a \emph{replacement path}. 
For a given integer $L$, let $SP^L_G(s,t,F)$ be the collection of all the shortest $s$-$t$ paths in $G \setminus F$ that contain \emph{at most} $L$ edges. A path in $SP^L_G(s,t,F)$ is referred to as $P^{L}_G(s,t,F)$. Let $d^{L}(s,t,G\setminus F)=\len(P^{L}_G(s,t,F))$.  If $SP^L_G(s,t,F)=\emptyset$, i.e.,  there is no path from $s$ to $t$ in $G \setminus F$ containing at most $L$ edges, then define $P^{L}_G(s,t,F)=\emptyset$ and $d^{L}(s,t,G\setminus F)=\infty$. For $F=\emptyset$, we abbreviate $P^L_G(s,t,\emptyset)=P^L_G(s,t)$ as the shortest $s$-$t$ path with at most $L$ edges, and $d^{L}(s,t,G)=\len(P^L_G(s,t))$ is the length of the path.
When the graph $G$ is clear from the context, we may omit it and write $P(s,t,F)$ and $P^L(s,t,F)$.
%
%

%
%
%
The following lemma is obtained via the doubling method\footnote{The algorithm provided in \cite{yuster2005answering} is randomized and it is described how to derandomize it \emph{with essentially no loss in efficiency} in Sec 8 of  \cite{yuster2005answering}.}  of \cite{yuster2005answering}, recently used in \cite{ChechikC20}.
\begin{lemma}\label{lem:shortAPSP}[Lemma 5 of \cite{ChechikC20}]
For every $n$-vertex subgraph $G' \subseteq G$, there is an algorithm that computes $\{d^{L}(s,t,G'), P^{L}(s,t,G')\}_{s,t \in V}$ in time $\widetilde{O}(L M n^{\omega})$.
\end{lemma}

The next lemma summarizes the quality of the randomized $(L,f)$-\RP procedures as obtained in \cite{weimann2013replacement} and \cite{dinitz2011fault}.  The proof is deferred to Appendix \ref{sec:miss-proof}.
\begin{lemma}[Randomized $(L,f)$-\RP]\label{lem:rand-RPC}
For every $n$-vertex graph $G=(V,E)$ and integer parameters $L,f \leq n$, one can compute a collection $\mathcal{G}=\{G_1,\ldots, G_r\}$ of $r$ subgraphs such that w.h.p. $\mathcal{G}$ is an $(L,f)$-\RP, where 
$r=O(f\cdot \max\{L,f\}^{\min\{L,f\}} \cdot \log n)$. The computation time is $O(r \cdot |E|)$.
\end{lemma}
\def\APPENDRANDLFRPC{
\begin{proof}[Proof of Lemma \ref{lem:rand-RPC}]
First consider the case where $L\geq f$. Let $\mathcal{G}=\{G_1,\ldots, G_r\}$ be a collection of independently sampled subgraphs for $r=c \cdot f \cdot L^f \log n$ where $c$ is a sufficiently large constant. Each subgraph $G_i$ is obtained by sampling each edge $e \in E(G)$ into $G_i$ independently with probability $p=1-1/L$.   We now show that $\mathcal{G}$ is indeed an $(L,f)$-\RP.
Fix a replacement path $P(s,t,F)$ of length at most $L$ that avoids a set of $F$ edges. The probability that a subgraph $G_i$ covers $P(s,t,F)$ is at least $q=p^L \cdot 1/L^f=1/(e \cdot L^f)$. Thus the probability that none of the $r$ subgraphs covers $P(s,t,F)$ is at most $(1-q)^r\leq (1-1/(e \cdot L^f))^{c \cdot f \cdot L^f \log n}=1/n^{c' f}$ for a sufficiently large constant $1<c'<c$. By taking $c$ to be a sufficiently large constant, and applying the union bound over all $n^{4f+2}$ triplets of $s,t,F$, we get that w.h.p. $\mathcal{G}$ is an $(L,f)$-\RP.

Next, assume that $L\leq f$. The definition of $\mathcal{G}$ is almost the same up to a small modification in the selection of the parameters. Set $r=c \cdot f^{L+1} \log n$ and let $p=1/f$. 
To see the correctness, fix a replacement path $P(s,t,F)$ with at most $L$ edges. The probability that $G_i$ covers $P(s,t,F)$ is at least $q=p^L \cdot (1-p)^f=1/(e \cdot f^L)$. Thus the probability that none of the $r$ subgraphs covers $P(s,t,F)$ is at most $(1-q)^r\leq (1-1/(e \cdot f^L))^{c \cdot f^{L+1} \log n}=1/n^{c' f}$ for a sufficiently large constant $1<c'<c$. By taking $c$ to be a sufficiently large constant, and applying the union bound over all $n^{2f+2}$ triplets of $s,t,F$, we get that w.h.p. $\mathcal{G}$ is an $(L,f)$-\RP.
\end{proof}
}

\subsection{Error Correcting Codes}

In this subsection, we recall the definition of error correcting codes and some standard code constructions known in literature. We define below a notion of distance used in coding theory (called \emph{Hamming} distance) and then define error correcting codes with its various parameters. 

\begin{definition}[Distance]
Let $\Sigma$ be a finite set and $\ell\in\mathbb{N}$, then the distance\footnote{We use the normalized notion of distance for the sake of exposition. In coding theory literature, our notion of distance is referred to as \emph{relative distance}. } between $x,y\in \Sigma^\ell$, denoted by $\distc(x,y)$, is defined to be:
\[ \distc(x,y) = \frac{1}{\ell}\cdot \abs{\sett{i\in[\ell]}{x_i\neq y_i}}. \]
\end{definition}

\begin{definition}[Error Correcting Code]
Let $\Sigma$ be a finite set. For every $\ell\in\mathbb{N}$, a subset $C\subseteq \Sigma^\ell$ is said to be an error correcting code with block length $\ell$, message length $k$, and relative distance $\delta$ if $|C|\ge |\Sigma|^k$ and for every $x,y\in C$, $\distc(x,y)\geq \delta$. We denote then $\distc(C)=\delta$. Moreover, we say that $C$ is a \code{k,\ell,\delta}{q} code to mean that $C$ is a code defined over alphabet set of size $q$ and is of message length $k$, block length $\ell$, and relative distance $\delta$. Finally, we refer to the elements of a code $C$ as codewords. 
\end{definition}

For the results in this article, we require codes with certain extremal properties. First, we recall Reed-Solomon codes whose codewords are simply the evaluation of univariate polynomials over a finite field.

\begin{theorem}[Reed-Solomon Codes \cite{RS60}]\label{thm:rs}
For every prime power $q$, and every $k \leq  q$, there exists a \code{k, q, 1-\frac{k-1}{q}}{q} code. 
\end{theorem}

These codes achieve the best possible tradeoff between the rate of the code (i.e., the ratio of message length to block length) and the relative distance of the code in the large alphabet regime as they meet the Singleton bound \cite{Singleton}. However,  if we desire codes with alphabet size much smaller than the block length then, Algebraic-Geometric codes \cite{G70,TVZ82} are the best known construction of codes achieving a good tradeoff between rate and relative  distance (but do not meet the Singleton bound). We specify below a specific construction of such codes. 

\begin{theorem}[Algebraic-Geometric Codes \cite{GS96}] \label{thm:ag}
Let $p$ be a prime square greater than or equal to 49, and let $q:=p^c$ for any $c\in\mathbb{N}$. Then for every $k\in\mathbb{N}$, there exists a \code{k, k\cdot \sqrt{q}, 1-\frac{3}{\sqrt{q}}}{q} code. 
\end{theorem}

Finally, we recall here a well-known fact about code concatenation (for example see Chapter 10.1 of \cite{GRS19}).

\begin{fact}\label{fact:code}
Let $k,\ell_1,\ell_2,c,q\in\mathbb N$ and let $\delta_1,\delta_2\in[0,1]$. Suppose we are given a \code{k,\ell_1,\delta_1}{q^c} outer code $C_1$ and a \code{c,\ell_2,\delta_2}{q} inner code $C_2$. Then the concatenation of the two codes $C_1\circ C_2$ is a \code{k,\ell_1\cdot \ell_2,\delta_1\cdot \delta_2}{q} code.
\end{fact}

\section{Hit and Miss Hash Families}

In this section, we show the construction of a certain class of hash families which will subsequently be used to design a deterministic algorithm for computing an $(L,f)$-\RP with a small \CV.
Below we define the notion of Hit and Miss hash families.

\begin{definition}[Hit and Miss Hash Family] \label{def:HM}For every $N,a,b,\ell,q\in \mathbb{N}$ such that $b\le a$, we say that $\H:=\{h_i:[N]\to [q]\mid i\in[\ell]\}$ is a \code{N,a,b,\ell}{q}{}-Hit and Miss ($\mathsf{HM}$) hash family\footnote{The reasoning behind naming them as \emph{Hit and Miss} Hash Family is as follows. Fix $A$ and $B$.  There exists a hash function $h$ in the family and a subset $S$ of $[q]$ of size at most $b$  such that $S$ completely \emph{hits} $h(B)$ and completely \emph{misses} $h(A)$. All other interpretations of the name ``Hit and Miss'' Hash Family are for the entertainment of the reader.} if for every pair of mutually disjoint subsets $A,B$ of $[N]$, where $|A|\le a$ and $|B|\le b$, there exists some $i\in[\ell]$ such that:
\begin{align}
 \forall (x,y)\in A\times B,\ h_i(x)\neq h_i(y).\label{eq:HM}
\end{align}

In the cases when $N,a,b$ is clear from the context, we simply refer to $\H$ as a \code{\ell}{q}{}-\HM{}.  Moreover, the computation time of a \code{\ell}{q}{}-\HM{} is defined to be the time needed to output the $\ell\times N$ matrix with entries in $[q]$ whose $(i,x)^{\text{th}}$ entry is simply $h_{i}(x)$ (for $h_i\in\H$).
\end{definition}

We begin our discussion by noting that there exist  a naive \code{1}{N}{}-\HM{} and a naive \code{\binom{N}{\le b}}{2}{}-\HM{}. Our goal is to construct a \code{\ell}{2}{}-\HM with the smallest possible value for $\ell$, as this is important for the  applications in the future sections. Towards this goal we prove the theorem below.

\begin{theorem}[Small Boolean Hit and Miss Hash Family]\label{thm:HM}
Given  integers $N,a,b$ such that $b\le a$, there is a deterministic algorithm $\mathcal{A}$ for computing an     \code{N,a,b,\ell}{2}{}-$\mathsf{HM}$ hash family where:
  \[
    \ell\le \left\{\begin{array}{lr}
       (\alpha c ab)^{b+1}, & \text{if }a\ge N^{\nicefrac{1}{c}},\ \text{for some constant }c\in\mathbb{N},\\
        (\alpha ab)^{b+2}\cdot \log N, & \text{if } a=N^{o(1)}\text{ and }b=\Omega(\log N),\\
                (\alpha ab)^{b+2}\cdot \log N, & \text{if } a\le \log N,\\
       (\alpha ab\log N)^{b+1}, & \text{otherwise,}
        \end{array}\right.
  \]
for some small universal constant $\alpha\in\mathbb N$.  
Moreover, the running time of $\mathcal{A}$ denoted by $T(\mathcal{A})$ is,
$$
T(\mathcal{A})=
\left\{\begin{array}{lr}
 N^{1+o(1)}\cdot \ell&\text{if } a=N^{o(1)}\text{ and }b=\Omega(\log N),\\
N\cdot (\log N)^{O(1)}\cdot \ell, & \text{otherwise.}
        \end{array}\right.
$$
\end{theorem}

Note that the above theorem significantly improves on the naive \code{\binom{N}{\le b}}{2}{}-\HM whenever $ab\ll N$.
Before we formally prove the above theorem, let us briefly outline our proof strategy. Our approach is to start from the naive \code{1}{N}{}-\HM{} and first construct a \code{\ell}{q}{}-\HM{} (for some $q,\ell\in\mathbb N$) where we try to minimize the quantity $\binom{q}{\le b}\cdot \ell$ (which is roughly $q^b\cdot \ell$). The reason for minimizing $q^b\cdot \ell$ is because we show below how to start from a
\code{\ell}{q}{}-\HM{} and trade off the size of the range of the hash function for the size of the hash family, in order to obtain an \code{\ell\cdot \binom{q}{\le b}}{2}{}-\HM{}.

\begin{lemma}[Alphabet Reduction]\label{lem:alphabet}
Given integers $N,a,b,q,\ell$ such that $b\le a$, and  a \code{N,a,b,\ell}{q}{}-\HM{} $\H$, there exists a \code{N,a,b,\ell\cdot \binom{q}{\le b}}{2}{}-\HM{} $\H'$ which can be computed in time $O(q^b\cdot T_{\H})$, where $T_{\H}$ is the time needed to compute $\H$. 
\end{lemma}
\begin{proof}
Given $\H:=\{h_i:[N]\to[q]\mid i\in[\ell]\}$, we define $\H':=\{h'_{i,S}:[N]\to\{0,1\}\mid i\in~\ell, S\subseteq~[q], |S|\le b \}$ as follows:
$$
\forall (i,S)\in [\ell]\times \binom{[q]}{\le b}, \forall x\in[N],\ \  h'_{i,S}(x)=\begin{cases}
0\text{ if }h_i(x)\in S,\\
1\text{ otherwise.}
\end{cases}
$$

It is clear that there are $\ell\cdot \binom{q}{\le b}$ many hash functions in $\H'$, and therefore in order to show that $\H'$ is a \code{N,a,b,\ell\cdot \binom{q}{\le b}}{2}{}-\HM{}, it suffices to show that \eqref{eq:HM} holds. To see this fix any disjoint sets $A,B \subseteq [N]$ such that $|A|\leq a$ and $|B|\leq b$. Since $\H$ is a\code{N,a,b,\ell}{q}{}-\HM{}, there exists some $i^*\in[\ell]$ such that
\begin{align}\label{eqbool}
\forall (x,y)\in A\times B,\ \text{we have }h_{i^*}(x)\neq h_{i^*}(y).
\end{align}   
Consider the subset $S^*:=\{h_{i^*}(y)\mid y\in B\}$. Clearly $|S^*|\le |B|\le b$. Therefore we have that for every $y\in B$, $h'_{i^*,S^*}(y)=0$. On the other hand from \eqref{eqbool}, we have that for all $x\in A$, $h_{i^*}(x)\notin S^*$. Therefore, for every $x\in A$, $h'_{i^*,S^*}(x)=1$.   Thus we have established \eqref{eq:HM}.
The computation time of $\H'$ follows from noting that $\binom{q}{\le b}\le (1+q)^b$.
\end{proof}

As a simple demonstration of how we will use the above lemma, notice that if we combine the above lemma with the naive \code{1}{N}{}-\HM, then we obtain the \code{\binom{N}{\le b}}{2}{}-\HM. 

Following the proof strategy we mentioned before the statement of Lemma~\ref{lem:alphabet}, we focus now on constructing non-trivial \code{\ell}{q}{}-\HM, with the goal of minimizing the quantity $\binom{q}{\le b}\cdot \ell$. As a warm up, we show below a simple construction that achieves very good parameters.

\begin{lemma}\label{cl:small-perfect}
Given integers $N,a,b$ such that $b\le a$, there exists a \code{N,a,b,{1+ab\log N}}{O(ab(\log N)^2)}{}-\HM{}. 
\end{lemma}
\begin{proof}
The family $\mathcal{H}$ we consider consists of all functions $h_p(x)=x (~\mod ~p)$ for the first $1+ab\log N$ \emph{prime numbers} $p$. Note that the $(1+ab\log N)^{\text{th}}$ prime number is at most $1+2ab\log N(1+\log a+\log b + \log\log N)= O(ab(\log N)^2)$.   Thus, in order to show that $\H$ is a \code{N,a,b,1+ab\log N}{O(ab(\log N)^2)}{}-\HM{}, we just need to show \eqref{eq:HM}.
Fix two disjoint sets $A,B \subseteq [N]$ such that $|A|\leq a$ and $|B|\leq b$.  Consider the following quantity.

$$\alpha_{A,B}:=\prod_{x \in A, y \in B}|y-x|.$$

Note that since $|y-x|\in[0,N]$ for every $(x,y)\in A\times B$,  we have that $\alpha_{A,B}\le N^{ab}$. It is known that the product of the first $m$ primes (called primorial function) is upper bounded $e^{m(1+o(1))}$. Let $\alpha'\in[1, \alpha_{A,B}]$ be the number with the most number of prime factors. It is clear then that the number of prime factors of $\alpha'$ is the largest $m$, for which we have $e^{m(1+o(1))}\le \alpha'\le N^{ab}$. This implies $m\le ab\log N$. 
 Thus, $\alpha_{A,B}$ has at most $ab\log N$ distinct prime factors. Therefore, given any set of $1+ab\log N$ prime numbers there must exist a prime that does not divide $\alpha_{A,B}$. On the other hand note that
 for  $(x,y)\in A\times B$ and a prime $p$, we have that $x (~\mod~ p) = y (~\mod~ p)$ implies that $p$ divides $\alpha_{A,B}$. Thus, there must exist  a prime in the first $1+ab\log N$ prime numbers for which we have $x (~\mod~ p) \neq  y (~\mod~ p)$ for all  $(x,y)\in A\times B$.
\end{proof}

We remark the above proof strategy of using (modulo) prime numbers has been used many times in literature, for example \cite{alon1996derandomization}. Next, we show a systematic way to construct a \HM\ from error correcting codes and then use specific codes to improve on the parameters of the above lemma.

\begin{proposition}\label{prop:codehash}
Let $N,a,b,\ell\in\mathbb N$ and $\delta\in[0,1]$ such that $\delta> 1-\frac{1}{ab}$. 
Then, every \code{\log_q N,\ell,\delta}{q} code can be seen as a \code{N,a,b,\ell}{q}{}-\HM. 
\end{proposition}
\begin{proof}
Given a \code{\log_q N,\ell,\delta}{q} code $C$, where for every $i\in[N]$, $C(i)$ denotes the $i^{\text{th}}$ codeword (under some canonical labeling of the codewords of $C$), we define the hash family $\H:=\{h_i:[N]\to q\mid i\in[\ell]\}$ as follows:
$$
\forall i\in[\ell],\ \forall x\in[N],\ h_i(x)=C(x)_i,
$$
where $C(x)_i$ denotes the $i^{\text{th}}$ coordinate of $C(x)$ (i.e., the $i^{\text{th}}$ coordinate of the $x^{\text{th}}$ codeword). To see that $\H$ is a \code{N,a,b,\ell}{q}{}-\HM, we need to show \eqref{eq:HM}. Fix disjoint $A,B\subseteq [N]$ where $|A|\le a$ and $|B|\le b$.  For every $(x,y)\in A\times B$ we have:
\begin{align}
\Pr_{i\sim[\ell]}\left[h_i(x)\neq h_i(y)\right]=\Delta(x,y)\ge \delta.\label{eqmin}
\end{align}

By a simple union bound we have that, 
\begin{align}
\Pr_{i\sim[\ell]}\left[\forall (x,y)\in A\times B,\ h_i(x)\neq h_i(y)\right]\ge 1-ab\cdot (1-\delta).\label{equnion}
\end{align}
Finally, \eqref{eq:HM} follows by noting that $\delta> 1-\frac{1}{ab}$.
\end{proof}

By a direct application of  the parameters of Reed-Solomon codes (Theorem~\ref{thm:rs})  to the above proposition we obtain the following.

\begin{corollary}[Reed-Solomon Hash Family]\label{cor:RS}
Given integers $N,a,b$ such that $b\le a$, there exists a \code{N,a,b,O\left(\frac{ab\log N}{\log a}\right)}{O\left(\frac{ab\log N}{\log a}\right)}{}-\HM{}.  Moreover, the computation time of the \HM is $O\left(abN(\log N)^2\right)$.
\end{corollary}
\begin{proof}
Let $q$ be the smallest prime greater than $\frac{ab\log N}{\log a}$  (note that $q\in\left(\frac{ab\log N}{\log a},\frac{2ab\log N}{\log a}\right)$). Let $C$ be the  \code{\log_q N, q, 1-\frac{\log_q N}{q}}{q} code guaranteed from Theorem~\ref{thm:rs}. From Proposition~\ref{prop:codehash} we can think of $C$ as a \code{N,a,b,q}{q}{}-\HM since 
$$\Delta(C)=1-\frac{\log N}{q\log q}> 1-\frac{\log N\log a}{ab\log N\log a}=1-\frac{1}{ab}.$$ 
By noting that $q< \frac{2ab\log N}{\log a}$, we may say that $C$ is a \code{N,a,b,O\left(\frac{ab\log N}{\log a}\right)}{O\left(\frac{ab\log N}{\log a}\right)}{}-\HM{}.

It is known that the generator matrix of Reed Solomon codes mentioned in Theorem~\ref{thm:rs} can be constructed in near linear time of the size of the generator matrix   \cite{RS60}. Once we are given the generator matrix of $C$, outputting any codeword  can be done in  $O(q\log\log N)$ time using Fast Fourier Transform. Therefore the computation of the corresponding \HM can be done in time $O(qN\log\log N)=O(abN\log N\log\log N)$.
\end{proof}

In fact, we obtain  a \code{\frac{1+ab\log N}{\log a+ \log b+\log\log N}}{\frac{1+ab\log N}{\log a+ \log b+\log\log N}}{}-\HM from Reed-Solomon codes but chose to write a less cumbersome version in the corollary statement. Note that while the size of the Hash families of Lemma~\ref{cl:small-perfect} and the above corollary are the same when $a\ll N^{o(1)}$, but even in that case we save a $\log N$ factor in the alphabet size of the hash function.

In order to explore further savings in the alphabet size of the hash function, we apply  the parameters of Algebraic-Geometric codes (Theorem~\ref{thm:ag}) to Proposition~\ref{prop:codehash}  and  obtain the following.

\begin{corollary}[Algebraic-Geometric Hash Family]\label{cor:AG}
Given integers $N,a,b$ such that $b\le a$, there exists a \code{O(ab\log N)}{O(a^2b^2)}{}-\HM.  Moreover, the computation time of the \HM is $O\left((ab\log N)^3 + Nab\log^3 N\right)$.
\end{corollary}
\begin{proof}
Let $p$ be the smallest prime greater than $3ab$  (note that $p\in(3ab,6ab)$) and let $q=p^2$. Let $C$ be the  \code{\log_q N, \sqrt{q}\cdot \log_q N, 1-\frac{3}{\sqrt{q}}}{q} code guaranteed from Theorem~\ref{thm:ag}. From Proposition~\ref{prop:codehash} we can think of $C$ as a \code{N,a,b,\sqrt{q}\cdot \log_q N}{q}{}-\HM since 
$$\Delta(C)=1-\frac{3}{p}> 1-\frac{1}{ab}.$$ 
By noting that $q\le 36a^2b^2$, we may say that $C$ is a \code{N,a,b,O\left(\frac{ab\log N}{\log a}\right)}{O(a^2b^2)}{}-\HM{}.

It is known that the generator matrix of Algebraic-Geometric codes mentioned in Theorem~\ref{thm:ag} can be constructed in near cubic time of the block length of the code  \cite{SAKSD01}. Therefore the computation of the corresponding \HM can be done in time $O((ab\log N)^3 + Nab\log^3 N)$.
\end{proof}

\begin{sloppypar} However these parameters are worse than the  parameters of Corollary~\ref{cor:RS} whenever $ab~\gg~\log N$. 
We construct below a specific code concatenation of Reed-Solomon codes and Algebraic-Geometric codes that does indeed improve on the parameters of Corollary~\ref{cor:RS} for the setting when $a,b$ are not too small. \end{sloppypar}

\begin{lemma}\label{lem:concat}
Let $p$ be a prime square greater than or equal to 49, and let $q:=p^c$ for any  $c\in\mathbb{N}$. Then for every $k\in\mathbb{N}$, there exists a \code{k, k\cdot q, 1-\frac{4}{\sqrt{q}}}{\sqrt{q}} code. 
\end{lemma}
\begin{proof}
We concatenate the \code{k, k\cdot \sqrt{q}, 1-\frac{3}{\sqrt{q}}}{q} code from Theorem~\ref{thm:ag} (treated as the outer code) with the \code{2,\sqrt{q},1-\frac{1}{\sqrt{q}}}{\sqrt{q}} code from Theorem~\ref{thm:rs} (treated as the inner code). From Fact~\ref{fact:code}, this gives us the desired code.
\end{proof}

It is worth noting that while concatenation codes obtained by combining Reed-Solomon codes and Algebraic-Geometric codes have appeared many times in literature, to the best of our knowledge, this is the first time that Algebraic-Geometric codes are the outer code and Reed-Solomon codes are the inner code (as Algebraic-Geometric codes are typically used for their small alphabet size).

An immediate corollary of Proposition~\ref{prop:codehash} and Lemma~\ref{lem:concat} is the following.

\begin{corollary}[Concatenated Hash Family]\label{cor:hash}
Given integers $N,a,b$ such that $b\le a$, there exists a \code{N,a,b,O\left(\frac{a^2b^2\log N}{\log a}\right)}{O(ab)}{}-\HM{}. 
Moreover, the computation time of the \HM is  $O\left(N\cdot (ab\log N)^3 \right)$.
\end{corollary}
\begin{proof}
Let $p$ be the smallest prime greater than $4ab$  (note that $p\in(4ab,8ab)$). Let $q:=p^2$ and $C$ be the  \code{\log_q N, q\cdot \log_q N, 1-\frac{4}{p}}{p} code guaranteed from Lemma~\ref{lem:concat}. From Proposition~\ref{prop:codehash} we can think of $C$ as a \code{N,a,b,q\cdot \log_q N}{p}{}-\HM since 
$$\Delta(C)=1-\frac{4}{p}> 1-\frac{1}{ab}.$$ By noting that $p\le 8ab$, we may say that $C$ is a \code{N,a,b,O\left(\frac{a^2b^2 \log N}{\log a}\right)}{O(ab)}{}-\HM.

It is known that the generator matrix of the codes mentioned in Theorem~\ref{thm:ag} (resp.\ Theorem~\ref{thm:rs}) can be constructed in cubic time in the block length of the code \cite{SAKSD01} (resp.\ linear time in the block length of the code \cite{RS60} as the message length is 2). Therefore the computation of the corresponding \HM can be done in time $O(N\cdot (ab\log N)^3 )$.
 \end{proof}

We finally wrap up by noting below  that the proof of Theorem~\ref{thm:HM} follows from combining Lemma~\ref{lem:alphabet} with  Corollaries~\ref{cor:RS}~and~\ref{cor:hash}. 

\begin{proof}[Proof of Theorem~\ref{thm:HM}]

\begin{sloppypar}Suppose $a\ge N^{\nicefrac{1}{c}}$, for some constant $c\in\mathbb{N}$ then consider the \code{O\left(\frac{ab\log N}{\log a}\right)}{O\left(\frac{ab\log N}{\log a}\right)}{}-\HM from Corollary~\ref{cor:RS} and note that $\frac{\log N}{\log a}\le c$.  Let the alphabet of this \HM be $\beta cab$, for some universal constant $\beta$. Then, we invoke Lemma~\ref{lem:alphabet} on this \code{O\left(cab\right)}{\beta cab}{}-\HM to obtain the desired Boolean \HM.  The computation time of the final \HM is $O\left((\beta cab)^b\cdot abN(\log N)^2\right)=O\left(N\cdot (\log N)^2\cdot (\beta cab)^{b+1}\right)$.\end{sloppypar}

\begin{sloppypar}Suppose $a=N^{o(1)}$ and $b=\Omega(\log N)$ (or suppose $a\le \log N$) then consider the \code{O\left(\frac{a^2b^2\log N}{\log a}\right)}{O\left(ab\right)}{}-\HM from Corollary~\ref{cor:hash} and ignore the $\log a$ term in the denominator in the expression for the size of the hash family. Let the alphabet of this \HM be $\beta' ab$, for some universal constant $\beta'$. Then, we invoke Lemma~\ref{lem:alphabet} on this \code{O\left(a^2b^2\log N\right)}{\beta' ab}{}-\HM to obtain the desired Boolean \HM.  The computation time of the final \HM is $O\left((\beta' ab)^b\cdot  N(ab\log N)^3\right)=O\left(N\cdot \log N\cdot (\beta' ab)^{b+2}\cdot (ab\cdot (\log N)^2)\right)$. Notice that if $a\le \log N$ then the expression $(ab\cdot (\log N)^2)$ is $O(\log ^4 N)$. Otherwise if $a=N^{o(1)}$ then the expression $(ab\cdot (\log N)^2)$ is still $N^{o(1)}$. \end{sloppypar}

In every other case,  consider the \code{O\left(\frac{ab\log N}{\log a}\right)}{O\left(\frac{ab\log N}{\log a}\right)}{}-\HM from Corollary~\ref{cor:RS} and ignore the $\log a$ term in the denominator in the expressions for both the size of the hash family and the alphabet size.  Let the alphabet of this \HM be $\beta'' ab\log N$, for some universal constant $\beta''$. Then, we invoke Lemma~\ref{lem:alphabet} on this \code{O\left(ab\log N\right)}{\beta'' ab\log N}{}-\HM to obtain the desired Boolean \HM.  The computation time of the final \HM is $O\left((\beta'' ab\log N)^b\cdot  Nab\cdot (\log N)^2\right)=O\left(N\cdot \log N\cdot (\beta'' ab\log N)^{b+1}\right)$.
\end{proof}

In order to facilitate the applications in the next section we  introduce the notation \HMC{C} to denote the following: given a code $C$, we first interpret it as a \HM in accordance with Proposition~\ref{prop:codehash} and then apply Lemma~\ref{lem:alphabet} to this hash family to obtain a Boolean \HM, denoted by \HMC{C}{}.

\paragraph{Optimaility of Reed-Solomon based \HM.}
We digress for a short discussion on the optimality of the parameters of \HM constructed from Reed-Solomon codes. There are two reasons why one might suspect that the parameters of Corollary~\ref{cor:RS} can be improved. First is the union bound applied in \eqref{equnion}. Second is the bounding of the number of disagreements between two codewords by the relative distance  in \eqref{eqmin}. It seems intuitively not reasonable that there exists two subsets of codewords say $A$ and $B$ such that for every pair of codewords in $A\times B$ there is a \emph{unique} set of coordinates on which they agree. Additionally, the expected fraction of disagreements between any two Reed-Solomon codewords is $1-1/q$ and instead bounding it by the relative distance, particularly when we are taking an union bound later in \eqref{equnion}, seems to raise concerns if the analysis has slacks. Therefore we ask:

\begin{open}
Let $a,b,d\in \mathbb{N}$. What is the smallest prime $q$ such that the following holds? For every two disjoint subsets of degree $d$ polynomials over $\mathbb{F}_q$, denoted by $A$ and $B$, where $|A|=a$ and $|B|=b$, we have that there exists some $\alpha\in\mathbb{F}_q$ such that no pair of polynomials in $A\times B$ evaluate to the same value at $\alpha$. 
\end{open}

Clearly, from Proposition~\ref{prop:codehash}, we have that if $q$ is at least $dab+1$ then it suffices. But can we get away with a smaller value of $q$?

\paragraph{Perfect Hash Families.} We conclude the discussion on \HM by noting the connection between \HM and the notion of  Perfect hash families that has received considerable attention in literature (for example see \cite{FK84,FKS84,SS90, AABCHNS92,N94,alon1995color,NSS95,alon1996derandomization,FN01,AG10}). If we replace \eqref{eq:HM} in Definition~\ref{def:HM} with 
\begin{align}
 \forall (x,y)\in S, x\neq y,\ \text{we have }h_i(x)\neq h_i(y),
\end{align}
where $S\subseteq [N]$ then it conincides with the notion of perfect hash families.  In other words, \HM can be as a \emph{bichromatic} variant of perfect hash families. Indeed a connection between error correcting codes and perfect hash families (much like Proposition~\ref{prop:codehash}) was already known in literature \cite{A86}. We also remark that construction of perfect hash families based on AG codes was also known in literature \cite{WX01}, but to the best of our knowledge, construction of hash families based on the concatenated AG codes (with the specific parameters of   Lemma~\ref{lem:concat}) is a novel contribution of this paper.

Additionally, one may see the randomized construction of \RP in Lemma~\ref{lem:rand-RPC} as coloring each edge with a random color in $[L]$ if $L\ge f$ (resp.\ in $[f]$ if $f\ge L$) and then randomly choosing one of the colors in $[L]$ (resp.\ $[f]$) and deleting (resp.\ retaining) all the edges corresponding to that color. The randomized procedure stated in the above way is very closely related to the celebrated color coding technique \cite{alon1995color} and a well-known way to derandomize the color coding technique is via perfect hash functions. However, using the derandomization objects developed for color coding yields \HM  with suboptimal parameters as they do not use the product structure of the constraints given in the definition of \HM. Consequently, they lead to worse constructions than the ones we give in this paper (to see this set $a\gg b$ and note that $ab \ll (a+b)^2$). The use of $k$-restriction sets 
\cite{AlonMS06} also yields \HM with suboptimal parameters for the same reason. 

\subsection{Strong Hit and Miss Hash Families}

In order to have certain applications, we  introduce the following strengthening of Definition~\ref{def:HM}.

\begin{definition}[Strong Hit and Miss Hash Family] For every $N,a,b,\ell,q\in \mathbb{N}$ such that $b\le a$, we say that $\H:=\{h_i:[N]\to [q]\mid i\in[\ell]\}$ is a \code{N,a,b,\ell}{q}{}-Strong Hit and Miss ($\mathsf{SHM}$) hash family if for every pair of mutually disjoint subsets $A,B$ of $[N]$, where $|A|\le a$ and $|B|\le b$, we have:
\begin{align}
\Pr_{i\sim[\ell]}\left[\forall (x,y)\in A\times B,\ h_i(x)\neq h_i(y)\right]\ge \frac{1}{2}.\label{eq:SHM}
\end{align}

In the cases when $N,a,b$ is clear from the context, we simply refer to $\H$ as a \code{\ell}{q}{}-Strong \HM{}.  
\end{definition}

Similar to Corollaries~\ref{cor:RS}~and~\ref{cor:hash}, we can prove the following bounds for Strong \HM.

\begin{lemma}\label{lem:strongcodes}
Given integers $N,a,b$ such that $b\le a$, there exists:
\begin{description}
\item[Reed-Solomon Strong \HM]  a \code{N,a,b,O\left(\frac{ab\log N}{\log a}\right)}{O\left(\frac{ab\log N}{\log a}\right)}{}-Strong \HM{} whose computation time is $O\left(abN(\log N)^2\right)$.
\item[Algebraic-Geometric Strong \HM] a \code{N,a,b,O\left(\frac{a^2b^2\log N}{\log a}\right)}{O(ab)}{}-Strong \HM{} whose computation time   is  $O\left(N\cdot (ab\log N)^3 \right)$.
\end{description}
\end{lemma}
\begin{proof}[Proof Sketch]
The proof follows by noting the following. First, Proposition~\ref{prop:codehash} can be strengthened to say that if $\delta\ge 1-\frac{1}{2ab}$ then every \code{\log_q N,\ell,\delta}{q} code can be seen as a \code{N,a,b,\ell}{q}{}-Strong \HM. Second, Corollaries~\ref{cor:RS},~\ref{cor:AG},~and~\ref{cor:hash} can be modified to yield Strong \HM (instead of just \HM), by simply choosing the alphabet value of the underlying code currently in the proofs to be at least twice (for Reed Solomon codes) or four times (for AG codes concatenated with Reed Solomon codes) as large as what is currently written.  
\end{proof}

%
%
%

In order to facilitate the applications in the next section we  introduce the notation \SHMC{C} to denote the following: given a code $C$, we first interpret it as a Strong \HM and then apply Lemma~\ref{lem:alphabet} to this hash family to obtain a Boolean Strong \HM, denoted by \SHMC{C}{}.

\section{$(L,f)$-\RPC}

Equipped with the construction of Boolean Hit and Miss hash families from the previous section, we  show in this section  how to use them in order to efficiently construct \RP.

\begin{proposition}\label{prop:RPC}\begin{sloppypar}
Given a graph $G$ on $m$ edges and integer parameters $L,f$, and a \code{m,\max\{L,f\},\min\{L,f\},\ell}{2}{}-\HM $\H$, we can construct an $(L,f)$-\RP of $G$ denoted by $\mathcal{G}_{L,f}^{\H}$ such that \CV$(\mathcal{G}_{L,f}^{\H})=2\cdot \ell$. Moreover, the construction of $\mathcal{G}_{L,f}^{\H}$ can be done in time $O(m\ell+T_\H)$, where $T_\H$ is the computation time of $\H$. \end{sloppypar}
\end{proposition}
\begin{proof}
 Label the edges of $G$ using $[m]$. For every $(i,\rho)\in [\ell]\times\{0,1\}$, we construct a subgraph $G_{i,\rho}$ of $G$ as follows: for every $x\in[m]$, the edge with label $x$ in $G$ is retained in $G_i$ if and only if $h_i(x)=\rho$. Then $\mathcal{G}_{L,f}^{\H}$ is simply $\{G_{i,\rho}\mid i\in[\ell],\rho\in\{0,1\}\}$. 

To see that $\mathcal{G}_{L	,f}^{\H}$ is an $(L,f)$-\RP, fix any vertex pair $(s,t)\in V(G)\times V(G)$ and fix any fault set $F:=\{e_{r_1},\ldots ,e_{r_d}\}\subseteq E(G)$ where $|F|=d\le f$. Let $P(s,t,F)$ be a replacement path with at most $L$ edges, i.e., $P(s,t,F)=\{e_{j_1},\ldots ,e_{j_t}\}\subseteq E(G)$, where $t\le L$. Consider the following two subsets of $[m]$: $A=\{j_1,\ldots ,j_t\}$ and $B=\{r_1,\ldots ,r_d\}$. Note that since  $P(s,t,F)$ is a replacement path we have $A$ and $B$ are disjoint subsets of $[m]$. From \eqref{eq:HM} we have that there exists some $i^*\in[\ell]$ such that for all $(x,y)\in A\times B$ we have $h_{i^*}(x)\neq h_{i^*}(y)$. Therefore we have that if $h_{i^*}(j_1)=0$ (resp.\ if $h_{i^*}(j_1)=1$) then in the graph $G_{i^*,0}$ (resp.\ $G_{i^*,1}$), we have that all edges of $P(s,t,F)$ are present and all edges of $F$ are absent.

In order to justify the computation time of $\mathcal{G}_{L,f}^{\H}$, we first compute the $\ell\times m$ Boolean matrix  $M_{\H}$ corresponding to $\H$ where the $(i,x)^{\text{th}}$ entry of $M_{\H}$ is simply $h_{i}(x)$.  After the computation of $M_{\H}$ we simply go over each row of the matrix to build the subgraphs. 
\end{proof}

\begin{proof}[Proof of Theorem~\ref{thm:general-covering}]
The proof follows immediately by putting together Theorem~\ref{thm:HM} with Proposition~\ref{prop:RPC} and noting that for every \code{N,a,b,\ell}{2}{}-\HM $\H$ used in  Theorem~\ref{thm:HM}, we have $T_{\H}>N\cdot \ell$.
\end{proof}

\begin{remark}
\label{rem:app}
For all the applications in this paper, we never use the construction of $(L,f)$-\RP given in Theorem~\ref{thm:general-covering} when $a=m^{o(1)}$ and $b=\Omega(\log m)$ (mainly because it has a prohibitive run time), and the result for that regime is merely of interest for bounding the covering number. \end{remark}

\paragraph{Useful properties of  $(L,f)$-\RP when $L \geq f$.}
A crucial property of the $(L,f)$-\RP that is needed for applications in the future section is that for every fixed set of faults $F$ there will be only a very small set of subgraphs in the covering set $\mathcal{G}_{L,f}$ that avoid $F$. 
As we see below, we have that the construction of $(L,f)$-\RP of Theorem \ref{thm:general-covering}  gives this additional property for free.


\begin{theorem}
\label{lem:num-collide-faults}
Let $L\ge f$ and $f=o(\log m)$, then one can compute an $(L,f)$-\RP $\mathcal{G}_{L,f}$  with the same \CV\ and time bounds
as in Theorem~\ref{thm:general-covering} that in addition satisfies the following property.  Let $F$ be a set of $d\leq f$ edge failures. Then, there exist a collection $\mathcal{G}_{F}$ of at most $fL\cdot\polylog(m)$ subgraphs in $\mathcal{G}_{L,f}$ that satisfy the following: 
\begin{itemize}
\item Every subgraph in $\mathcal{G}_{F}$ does not contain any of the edges in $F$.
\item For every vertex pair $(s,t)$ and every $P(s,t,F)$ path of length at most $L$, there exists a subgraph $G' \in \mathcal{G}_{F}$ that contains $P(s,t,F)$. 
\end{itemize}
Finally, given $F$ and $\mathcal{G}_{L,f}$, one can detect the subgraphs in $\mathcal{G}_{F}$ in time $fdL\cdot \polylog(m)$. 
\end{theorem}

The proof of the above theorem follows by the more general statement below about code based constructions of  \HM, and applying to it the parameters of specific codes. 

\begin{lemma}\label{lem:code-to-subgraphs}
Given a graph $G$ on $m$ edges and integer parameters $L,f,q,\ell$, and a  \code{\log_q m,\ell,\delta}{q} code $C$ with relative distance $\delta> 1-\frac{1}{Lf}$, then,
the $(L,f)$-\RP $\mathcal{G}_{L,f}$ given by Proposition~\ref{prop:RPC} on providing \HMC{C} has the following property. Let $F$ be a set of $d\leq f$ edge failures. Then, there exist a collection $\mathcal{G}_{F}$ of at most $\ell$ subgraphs in $\mathcal{G}_{L,f}$ that satisfy the following: 
\begin{itemize}
\item Every subgraph in $\mathcal{G}_{F}$ does not contain any of the edges in $F$.
\item For every vertex pair $(s,t)$ and every $P(s,t,F)$ path of length at most $L$, there exists a subgraph $G' \in \mathcal{G}_{F}$ that contains $P(s,t,F)$. 
\end{itemize}
Moreover, given $F$ and $\mathcal{G}_{L,f}$, one can detect the subgraphs in $\mathcal{G}_{F}$ in time $O(d\cdot (\ell+\mathsf{ev}(C)))$, where $\mathsf{en}(C)$ is the time needed to encode a message using $C$. 
\end{lemma}
\begin{proof}
 For every $i\in [\ell]$ let $S_i\subseteq [q]$ be defined as:
$$
S_i:=\{C(r_j)_i\mid j\in[d]\},
$$ 
where $F=\{e_{r_1},\ldots ,e_{r_d}\}$, and $C(r_j)_i$ is the $i^{\text{th}}$ coordinate of the $r_j^{\text{th}}$ codeword of $C$. For every $i\in[\ell]$ we include the subgraph $G_i$ in $\mathcal{G}_F$ if and only if the only edges in $G$ removed in $G_i$ are the ones mapped to an element of $S_i$ under $C_i$. It is clear that $|\mathcal{G}_F|$ by definition is at most $\ell$. Moreover, 
the computation time of the indices of the graphs in $\mathcal{G}_F$ is  $O(d\cdot (\ell+\mathsf{en}(C)))$ as once we encode the $d$ edges of $F$ using $C$, we can specify the indices of the subgraphs in $\mathcal{G}_F$ explicitly as defined above.

To note that $\mathcal{G}_F$ is a subset of $\mathcal{G}_{L,f}$, notice that for every $i\in [\ell]$ and every $S_i$ as defined above, we have in \HMC{C} a hash function $h:[m]\to\{0,1\}$ which maps to 0 exactly those edges (labels of edges) whose corresponding codeword on the $i^{\text{th}}$ coordinate is contained in $S_i$ (see the proof of Lemma~\ref{lem:alphabet} to verify this). Then, whence   \HMC{C} is provided to Proposition~\ref{prop:RPC}, the graph $G_{i,1}$ in $\mathcal{G}_{L,f}^{\text{\HMC{C}}}$ in the proof of Proposition~\ref{prop:RPC} is precisely the graph $G_i$ in $\mathcal{G}_F$.

All that is left to show are the structural properties of $\mathcal{G}_F$. By definition of $G_i$, it is clear that all the edges in $F$ are removed in each $G_i$. Furthermore, for every vertex pair $(s,t)\in V(G)\times V(G)$ and every replacement path $P(s,t,F)=\{e_{j_1},\ldots ,e_{j_t}\}$ with at most $L$ edges, we have from \eqref{eq:HM} that there is some $i^*\in[\ell]$ such that for all $\kappa\in[t]$, we have  $C({j_{\kappa}})_{i^*}\notin S_{i^*}$ (i.e., we apply Proposition~\ref{prop:codehash} on $C$ to obtain a \HM and use \eqref{eq:HM} with $A=\{{j_1},\ldots ,{j_t}\}$ and $B=\{{r_1},\ldots ,{r_d}\}$). Therefore all the edges of $P(s,t,F)$ are retained in $G_{i^*}$. 
\end{proof}
 
\begin{proof}[Proof of Theorem~\ref{lem:num-collide-faults}]
Since we have $L\ge f$ and $f=o(\log m)$, the bounds in Theorem~\ref{thm:general-covering} follow here as well with setting $a=L$ and $b=f$, while avoiding the case when $b=\Omega(\log m)$. In order to see that the additional property holds, we only need to verify that for the Reed Solomon code $\rs$ and the concatenated code $\con$ (from Lemma~\ref{lem:concat}) when we plug in \HMC{\rs} and \HMC{\con} respectively into Lemma~\ref{lem:code-to-subgraphs}, that the parameters are as claimed in the theorem statement. 

The block length $\ell$ of $\rs$ is set to be at most $\frac{2Lf\log m}{\log L}$ in Corollary~\ref{cor:RS}. If $L\ge m^{1/c}$ then $|\mathcal{G}_F|\le\ell=O(cLf)$ and otherwise  we have $|\mathcal{G}_F|\le \ell=O(Lf\log m)$.  

The block length $\ell$ of $\con$ is set to be at most $\frac{64L^2f^2\log m}{\log L}$ in Corollary~\ref{cor:hash}. Since we apply this bound to the case where $L\le \log m$ then $|\mathcal{G}_F|\le \ell=O(L^2f^2\log m)=O(Lf\log^3 m)$.

Plugging in the bound on the above block lengths of the two codes into Lemma~\ref{lem:code-to-subgraphs} gives the bounds of the additional property in the theorem statement. Note that the encoding time of $\rs$ is $\ell\cdot \polylog(m)=Lf\polylog(m)$ and while the encoding time of a codeword $\con$ is $O(\ell^3)$, since $f\le L\le \log m$, we have that the encoding time of $\con$ is also $\polylog(m)$. 
\end{proof}

\paragraph{Useful properties of  $(L,f)$-\RP s when $L \leq f$.}
Parts of the next theorem maybe  morally seen as the analog of Theorem~\ref{lem:num-collide-faults}, only that for the setting of $L\leq f$, we bound the number of subgraphs that fully contain a given path segment with at most $L$ edges.

\begin{theorem}
\label{lem:num-collide-paths}
Let $L\le f$, then one can compute an $(L,f)$-\RP $\mathcal{G}_{L,f}$ with the same \CV\ and time bounds\footnote{\label{note1}We recall Remark~\ref{rem:app} to say that when $a=m^{o(1)}$ and $b=\Omega(\log m)$ in the statement of Theorem~\ref{thm:general-covering}, the covering number we aim to achieve is $(\alpha Lf\log m)^{b+1}$ instead of $(\alpha Lf)^{b+2}\cdot \log m$.}
as in Theorem~\ref{thm:general-covering} that in addition satisfies the following property.  
Let $P$ be a replacement path segment of at most $L$ edges. Then, there exist a collection $\mathcal{G}_{P}$ subgraphs in $\mathcal{G}_{L,f}$ that satisfy the following: 
\begin{description}
\item{(I1)} $|\mathcal{G}_{P}|=fL\cdot\polylog(m)$.
\item{(I2)} Given $P$ and $\mathcal{G}_{L,f}$, one can detect the subgraphs in $\mathcal{G}_{P}$ in time $fdL\cdot \polylog(m)$.
\item{(I3)} Every subgraph in $\mathcal{G}_{P}$ fully contains $P$.
\item{(I4)} For every set $F \subseteq E$ of at most $f$ edges, there are at least $|\mathcal{G}_{P}|/2$ subgraphs in $\mathcal{G}_{P}$ that fully avoid $F$.
\item{(I5)} Every subgraph in $\mathcal{G}_{P}$ has at most $\frac{m}{f}$ many edges.  
\item{(I6)} Computing the subset of edges  in each $G_i \in \mathcal{G}_{L,f}$ takes $\widetilde{O}(\frac{m}{f})$ time. 
\end{description}
Additionally, (I5) and (I6)  when applied to the vertex variant \RP $\mathcal{G}_{L,f}^v$  over a graph $G$ on $n$  vertices  with vertex fault parameter $f$ yield the following: (I5v) Every subgraph in $\mathcal{G}_{P}$ has at most $\frac{n}{f}$ many vertices and (I6v) computing the subset of vertices in each $G_i \in \mathcal{G}_{L,f}^v$ takes $\widetilde{O}(\frac{n}{f})$ time.
\end{theorem}

The proofs of (I1) to (I4)  of the above theorem follow by the more general statement below about code based constructions of Strong \HM, and applying to it the parameters of specific codes. The proofs of (I5) and (I6) follows by a nice property of linear codes.

\begin{lemma}\label{lem:code-to-paths}
Given a graph $G$ on $m$ edges and integer parameters $L,f,q,\ell$, and a  \code{\log_q m,\ell,\delta}{q} code $C$ with relative distance $\delta> 1-\frac{1}{2Lf}$, then,
the $(L,f)$-\RP $\mathcal{G}_{L,f}$ given by Proposition~\ref{prop:RPC} on providing \SHMC{C} has the following property. Let $P$ be a replacement path segment of $d\le L$ edges. Then, there exist a collection $\mathcal{G}_{P}$ of at most $\ell$ subgraphs in $\mathcal{G}_{L,f}$ that satisfy the following: 
\begin{itemize}
\item Every subgraph in $\mathcal{G}_{P}$ fully contains $P$.
\item For every set $F \subseteq E$ of at most $f$ edges, there are at least $|\mathcal{G}_{P}|/2$ subgraphs in $\mathcal{G}_{P}$ that fully avoid $F$.
\end{itemize}
Moreover, given $P$ and $\mathcal{G}_{L,f}$, one can detect the subgraphs in $\mathcal{G}_{P}$ in time $O(d \cdot (\ell+ \mathsf{en}(C)))$, where $\mathsf{en}(C)$ is the time needed to encode a message using $C$. 
\end{lemma}
\begin{proof}
 For every $i\in [\ell]$ let $S_i\subseteq [q]$ be defined as:
$$
S_i:=\{C(r_j)_i\mid j\in[d]\},
$$ 
where $P=\{e_{r_1},\ldots ,e_{r_d}\}$. For every $i\in[\ell]$ we include the subgraph $G_i$ in $\mathcal{G}_F$ if and only if the only edges in $G$ preserved in $G_i$ are the ones mapped to an element of $S_i$ under $C_i$. It is clear that $|\mathcal{G}_F|$ by definition is at most $\ell$. 
Moreover, 
the computation time of the indices of the graphs in $\mathcal{G}_P$ is  $O(d\cdot (\ell+\mathsf{en}(C)))$ as once we encode the $d$ edges of $P$ using $C$, we can specify the indices of the subgraphs in $\mathcal{G}_P$ explicitly as defined above.

To note that $\mathcal{G}_P$ is a subset of $\mathcal{G}_{L,f}$, notice that for every $i\in [\ell]$ and every $S_i$ as defined above, we have in \SHMC{C} a hash function $h:[m]\to\{0,1\}$ which maps to 0 exactly those edges (labels of edges) whose corresponding codeword on the $i^{\text{th}}$ coordinate is contained in $S_i$ (see the proof of Lemma~\ref{lem:alphabet} to verify this). Then, whence   \SHMC{C} is provided to Proposition~\ref{prop:RPC}, the graph $G_{i,0}$ in $\mathcal{G}_{L,f}^{\text{\SHMC{C}}}$ in the proof of Proposition~\ref{prop:RPC} is precisely the graph $G_i$ in $\mathcal{G}_P$.

All that is left to show are the structural properties of $\mathcal{G}_P$. By definition of $G_i$, it is clear that all the edges in $P$ are preserved in each $G_i$. Furthermore, for every set $F:=\{e_{j_1},\ldots ,e_{j_t}\} \subseteq E$ of at most $f$ edges, we have from \eqref{eq:SHM} that 
$$\Pr_{i\sim[\ell]}\left[\forall (x,y)\in [d]\times [t],\ C(e_{r_x})_i\neq C(e_{j_y})_i\right]\ge \frac{1}{2}.$$

Therefore all the edges of $F$ are avoided  in at least half the graphs in $\mathcal{G}_{P}$.
\end{proof}

\begin{proof}[Proof of Theorem~\ref{lem:num-collide-paths}]
Since we have $L\le f$, the bounds in Theorem~\ref{thm:general-covering} follow here as well with setting $a=f$ and $b=L$, while we avoid the case when $b=\Omega(\log m)$ in order to get the right bounds (we consider this case to be covered by the `otherwise' case construction in Theorem~\ref{thm:general-covering}). In order to see that (I1) to (I4) holds, we only need to verify that for the Reed Solomon code $\rs$ and the concatenated code $\con$ (from Lemma~\ref{lem:concat}) when we plug in \SHMC{\rs} and \SHMC{\con} respectively into Lemma~\ref{lem:code-to-paths}, that the parameters are as claimed in the theorem statement. 

The block length $\ell$ of $\rs$ is set to be at most $\frac{4Lf\log m}{\log L}$ in Lemma~\ref{lem:strongcodes}. If $L\ge m^{1/c}$ then $|\mathcal{G}_P|\le\ell=O(cLf)$ and otherwise  we have $|\mathcal{G}_P|\le \ell=O(Lf\log m)$.  

The block length $\ell$ of $\con$ is set to be at most $\frac{256L^2f^2\log m}{\log L}$ in Lemma~\ref{lem:strongcodes}. Since we apply this bound to the case where $L\le \log m$ then $|\mathcal{G}_P|\le \ell=O(L^2f^2\log m)=O(Lf\log^3 m)$.

Plugging in the bound on the above block lengths of the two codes into Lemma~\ref{lem:code-to-paths} gives (I1) to (I4) in the theorem statement. Note that the encoding time of $\rs$ is $\ell\cdot \polylog(m)=Lf\polylog(m)$ and while the encoding time of a codeword $\con$ is $O(\ell^3)$, since $L\le f\le \log m$, we have that the encoding time of $\con$ is also $\polylog(m)$. 

Thus we now look towards proving (I5) and (I6). Notice that since $L\le f$, and the Boolean \HM provided to Proposition~\ref{prop:RPC} in the proof of Theorem~\ref{thm:general-covering} arises from the alphabet reduction of Lemma~\ref{lem:alphabet}, we know that we can even exclude all the subgraphs $G_{i,1}$ (for all $i\in[\ell]$) in the proof of Proposition~\ref{prop:RPC}, to only have $\ell$ many subgraphs in $\mathcal{G}_{L,f}$. We will use this simplification later in this proof.

In order to see that every subgraph in $\mathcal{G}_{L,f}$  has at most $\frac{m}{f}$ many edges (i.e., (I5)), we only need to verify that the Reed Solomon code $\rs$ and the concatenated code $\con$ (from Lemma~\ref{lem:concat})  are 1-wise independent: A code $C\subseteq [q]^\ell$ is said to be  1-wise independent if and only if for every $i\in[\ell]$ and every $\zeta\in[q]$ we have $$  \Pr_{x\sim C}[x_i=\zeta]= \frac{1}{q}.$$

Let us first see  why it suffices for (I5)   to show that $\rs$ and $\con$ are 1-wise independent. Given $L,f$, fix a code $C\in\{\rs,\con\}$ which optimizes the parameters of Theorem~\ref{thm:general-covering}. 
Fix a subgraph $G'$ in $\mathcal{G}_{L,f}$. By construction of $\mathcal{G}_{L,f}$ there exists $h\in $ \HMC{C}, such that the edge $e_i$ in $G$ is retained in $G'$ if and only if $h(i)=0$. Since the hash functions in \HMC{C} are indexed by the set $[\ell]\times \binom{[q]}{\le L}$ (where $q$ is the alphabet size and $\ell$ is the block length of $C$), let the index of $h$ be $(j,S)\in [\ell]\times \binom{[q]}{\le L}$. Notice that the number of edges in $G'$ is simply the subset $E'\subseteq [m]$ defined as $E'=\{x\in[m]\mid C(x)_j\in S\}$. However, since $C$ is 1-wise independent, we have that $\Pr_{x\sim C}[x_j\in S]= \frac{|S|}{q}$, and thus $|E'|=m\cdot |S|/q\le mL/q$. If $C=\rs$ then $q\ge Lf\log m$, and thus $|E'|\le m/(f\log m)$, and if $C=\con$ then $q\ge Lf$, and thus $|E'|\le m/f$.  This proves (I5).

We now return our focus to showing that $\rs$ and $\con$ are 1-wise independent. In fact we will show a stronger statement: every linear code $C$ is 1-wise independent. Let $A_{\ell\times \log_q m}:=(\vec{a_1},\ldots ,\vec{a_\ell})$ be the generator matrix of $C\subseteq [q]^\ell$. Then we can rewrite the claim of showing  1-wise independence as follows: for every $i\in[\ell]$ and every $\zeta\in[q]$ we have $$  \Pr_{y\sim [q]^{\log_q m}}[(Ay)_i=\zeta]= \frac{1}{q}.$$
We now rewrite $(Ay)_i$ as $\langle \vec{a_i},y\rangle$, and since $\vec{a_i}$ is not the zero vector the claim follows (by even just a simple induction argument on the dimension).

Now we show (I6). Fix some $G_i$ in $\mathcal{G}_{L,f}$.  By construction of $\mathcal{G}_{L,f}$ we may interpret the index $i$ as some $(j,S)\in [\ell]\times \binom{[q]}{\le L}$ such that the edge $e_x$ in $G$ is retained in $G_i$ if and only if $C(x)_j\in S$. Let $A_C:=(\vec{a_1},\ldots ,\vec{a_\ell})$ be the generator matrix of $C$. We can determine  the subset $T$ of $[q]^{\log_q m}$ defined as follows:
$$
T:=\{x\in [q]^{\log_q m}\mid \langle \vec{a_j},x\rangle\in S\}.
$$

Then interpretting $T$ as a subset of $[m]$ simply gives us the edge set of $G_i$. To compute $T$ efficiently, we first compute for every  $r\in [q]^{(\log_q m) -1}$ and every $z\in S$, the value: $$\alpha:=\left(z-\sum_{w=1}^{(\log_q m) -1}\left(\vec{a_j}(w)\cdot r_w\right)\right)\cdot (\vec{a_{j}}(\log_q m))^{^{-1}}.$$  Then we include the vector $(r,\alpha)\in [q]^{\log_q m}$ into T. Thus $T$ can be computed in time $\tilde{O}(m|S|/q)=\tilde{O}(mL/q)$. And as before if $C=\rs$ then $q\ge Lf\log m$, and thus $\tilde{O}(mL/q)=\tilde{O}(m/(f\log m))$, and if $C=\con$ then $q\ge Lf$, and thus $\tilde{O}(mL/q)=\tilde{O}(m/f)$.  This proves (I6). 
\end{proof}

\section{Lower Bounds for $(L,f)$-\RPC}\label{sec:LB}
In this section we provide a lower bound construction for the covering value of $(L,f)$-\RP and establish Theorem \ref{thm:lower-bound}. Our lower bound graph is based on a modification of the graph construction used to obtain a lower bound on the size of $f$-failure FT-BFS structures, defined as follows. 
\begin{definition}[FT-BFS Structures]\cite{ParterP16,parter2015dual}
Given a (possibly weighted) $n$-vertex graph $G=(V,E)$, a source vertex $s \in V$, and a bound $f$ on the number of (edge) faults $f$, a subgraph $H \subseteq G$ is an $f$-failure FT-BFS structure with respect to $s$ if
$\dist(s,t, H \setminus F)=\dist(s,t, G \setminus F) \mbox{~for every~} t \in V, F \subseteq E(G), |F|\leq f~.$
\end{definition}
FT-BFS structures were introduced by the second author and Peleg \cite{ParterP16} for the single (edge or vertex) failure. It was shown that for any unweighted $n$-vertex graphs and any source node $s$, one can compute an $1$-failure FT-BFS subgraph with $O(n^{3/2})$ edges. This was complemented by a matching lower bound graph. In \cite{parter2015dual}, the lower bound graph construction was extended to any number of faults $f$, which would serve the basis for our $(L,f)$-\RP lower bound argument. 
\begin{fact}\label{fc:ft-bfs}\cite{parter2015dual}
For large enough $n$, and $f\geq 1$, there exists an $n$-vertex graph $G^*_f$ and a source vertex $s$ such that any $f$-failure-BFS structure with respect to $s$ has $\Omega(n^{2-1/(f+1)})$ edges.
\end{fact}
In the high-level, the lower bound graph $G^*_f$ consists of a dense bipartite subgraph $B$ with $\Omega(n^{2-1/(f+1)})$ edges, and a collection of $\{s\} \times V$ paths, that serve as replacement paths from $s$ to all other vertices in $G$. The collection of paths are defined in a careful manner in a way that forces any $f$-failure FT-BFS for $s$ to include \emph{all} the edge of the bipartite graph $B$. To translate this construction into one that yields an $(L,f)$-\RP of large \CV, our key idea is to \emph{shortcut} the edge-length of $\{s\} \times V$ replacement paths of $G^*_f$ by means of introducing weights to the edges. As a result, we get a weighted graph $G^w_f$ whose all $\{s\} \times V$ replacement paths have at most $L$ edges for any given parameter $L \leq (n/f)^{1/(f+1)}$. By setting the weights carefully, one can show that any $f$-failure FT-BFS for the designated source $s$ must have $\Omega(L^f \cdot n)$ edges. To complement the argument, consider the optimal $(L,f)$-\RP $\mathcal{G}$ of minimal value for $G^w_f$. Since all the $\{s\} \times V$ paths are of length at most $L$, the replacement paths are resiliently covered by $\mathcal{G}$. This yields the following simple construction of $f$-failure FT-BFS $H \subseteq G$: Compute a shortest-path tree in each subgraph $G' \in \mathcal{G}$, and take the union of these subgraphs as the output subgraph $H$. 
Since this construction yields an $f$-failure FT-BFS with $O(|\mathcal{G}|n)$ edges, we conclude that $|\mathcal{G}|=\Omega(L^f\cdot n)$. We next explain this construction in details. 

In the next description, we use the notation of \cite{parter2015dual} and introduced several key adaptations along the way.  Our lower bound graph $G^w_f$ similarly to Fact \ref{fc:ft-bfs} is based on a graph $G_f(d)$ which is defined inductively. Note that whereas in \cite{parter2015dual}, the graph $G_f(d)$ is unweighted, for our purposes (making all replacement paths short in terms of number of edges) some edges will be given weights. For $f=1$, $G_1(d)$ consists of three components: (i) a set of vertices $U=\{u^1_1,\ldots, u^1_d\}$ connected by a path $P_1=[u^1_1,\ldots, u^1_d]$, (ii) a set of terminal vertices $Z=\{z_1,\ldots, z_d\}$, and (iii) a collection of $d$ edges $e^1_i$ of weight $w(e^1_i)=6+2(d-i)$ connecting $u^1_i$ and $z_i$ for every $i \in \{1,\ldots,f\}$. The vertex $r(G_1(d))=u^1_1$, and the terminal vertices of $Z$ are the \emph{leaves} of the graph denoted by $\Leaf(G_1(d))=Z$. 
Each leaf node $z_i \in \Leaf(G_1(d))$ is assigned a label based on a labeling function
$\LAB_1: \Leaf(G_1(d))  \to E(G_1(d))^1$. The label of the leaf corresponds to a set of edge faults under which
the path from root to leaf is still maintained. Specifically, $\LAB_1(z_i, G_1(d))=(u^1_i,u^1_{i+1})$ for $i \leq d-1$ and $\LAB_e(z_i, G_1(d))=\emptyset$. In addition, define
$P(z_i,G_1(d)) = P_1[\Root(G_1(d)),u^1_i] \circ Q^1_i$
to be the path from the root $u^1_1$ to the leaf $z_i$.

We next describe the inductive construction of the graph $G_{f}(d)=(V_{f}, E_{f})$, for every $f\ge 2$, given the graph $G_{f-1}(d)=(V_{f-1}, E_{f-1})$. The weights are introduced only in this induction step, i.e., for $f\geq 2$. 
The graph $G_{f}(d)=(V_{f}, E_{f})$ consists of the following components. 
First, it contains a path $P_f=[u^f_1, \ldots, u^f_d]$, where
the node $\Root(G_{f}(d))=u^f_1$ is fixed to be the root.
In addition, it contains $d$ disjoint copies of the graph $G'=G_{f-1}(d)$,
denoted by $G'_1, \ldots, G'_d$ (viewed by convention as ordered from left to right),
where each $G'_i$ is connected to $u^f_i$ by a collection of $d$
edges $e^f_i$, for $i \in \{1, \ldots, d\}$,
connecting the vertices $u^f_i$ with $\Root(G'_i)$.
The \emph{edge weight} of each $e^f_i$ is $w(e^f_i)=(d-i)\cdot \depth(G_{f-1}(d))$.
In the construction of \cite{parter2015dual}, each edge $e^f_i$ is replaced by a \emph{path} $Q^f_i$ of length $w(e^f_i)$. This is the only distinction compared to \cite{parter2015dual}. Note that by replacing a path $Q^f_i$ by a single edge $e^f_i$ of weight $|Q^f_i|$, the weighted length of the replacement paths would preserve but their length in terms in number of edges is considerably shorter. The leaf set of the graph $G_{f}(d)$ is the union of the leaf sets of $G'_j$'s, $\Leaf(G_{f}(d))=\bigcup_{j=1}^d \Leaf(G'_j)$. See Fig. \ref{fig:lowerbound} for an illustration for the special case of $f=2$. 

Finally, it remains to define the labels $\LAB_f(z_i)$ for each $z_i \in \Leaf(G_{f}(d))$.
For every $j \in \{1, \ldots, d-1\}$ and any leaf $z_j \in \Leaf(G'_j)$,
let $\LAB_f(z_j, G_{f}(d))=(u^f_j,u^f_{j+1}) \circ \LAB_{f-1}(z_j, G'_j)$.
Denote the size (number of nodes) of $G_f(d)$ by $\NodesIn(f,d)$,
its depth (maximal weighted distance between two nodes) by $\depth(f,d)$,
and its number of leaves by  $\NLeaf(f,d) = |\Leaf(G_f(d))|$.
Note that for $f=1$, $\NodesIn(1,d) = 2d+\sum_{i=1}^d 4+2 \cdot (d-i) \leq 7d^2$,
$\depth(1,d)=6+2(d-1)$ (corresponding to the length of the path $Q^1_1$),
and $\NLeaf(1,d)=d$. Since in our construction, we only shortcut the length of the paths, the following inductive relations hold as in \cite{parter2015dual}.
\begin{observation}[Observation 4.2 of \cite{parter2015dual}]
\label{obs:rel}~
\begin{description}
\item{(a)}
$\depth(f,d)=O(d^f)$.
\item{(b)}
$\NLeaf(f,d)=d^f$.
\item{(c)}
$\NodesIn(f,d)=c \cdot d^{f+1}$ for some constant $c$.
\end{description}
\end{observation}
Consider the set of $\lambda=\NLeaf(f,d)$ leaves in $G_f(d)$, $\Leaf(G_f(d)) = \bigcup_{i=1}^d \Leaf(G'_i) = \{z_1, \ldots, z_\lambda\}$, ordered from left to right according to their appearance in $G(f,d)$. 
\begin{lemma}[Slight modification of Lemma 4.3 of \cite{parter2015dual}]
\label{lem:prop_induc_path}
For every $z_j$ it holds that: \\
(1) The path $P(z_j, G_f(d))$ is the only $u^f_1-z_j$ path in $G_f(d)$.\\
(2) $P(z_j, G_f(d)) \subseteq G \setminus \LAB_f(z_j, G_f(d))$.\\
(3) $P(z_i, G_f(d)) \not\subseteq G \setminus \LAB_f(z_j, G_f(d))$
for every $i>j$.\\
(4) $\len(P(z_i, G_f(d))) > \len(P(z_j, G_f(d)))$ for every $i < j$. 
\end{lemma}
In Lemma 4.3 of \cite{parter2015dual}, the forth claim discusses the length of the paths $P(z_i, G_f(d))$. In our case, since we shortcut the path by introducing an edge weight the equals to the length of the removed sub-path, the same claim holds only for the \emph{weighted length} of the path. 
We next show that thanks to our modifications the hop-diameter (i.e., measured by number of edges) of $G_{f-1}(d)$ is bounded, and consequently, all $\{s\}\times V$ replacement paths are short.  
\begin{claim}\label{cl:hop} 
The hop-diameter of $G_{f}(d)$ is $O(f \cdot d)$. 
\end{claim}
\begin{proof}
The claim is shown by induction on $f$. For $f=1$, the hop-diameter of $G_1(d)$ is $|P_1|=d$. 
Assume that the claim holds up to $f-1$ and that the hop-diameter of $G_{f-1}(d)$ is at most $(f-1)d$. 
The graph $G_f(d)$ is then connected to $G_{f-1}(d)$ via the path $P_f=[u^f_1,\ldots, u^f_d]$ of hop-length $d$.
Each $u^f_i$ is connected to the root of the $i$th copy of $G_{f-1}(d)$ via an edge. Thus the hop-diameter of $G_{f}(d)$ is at most $f \cdot d$.
\end{proof}
Finally, we turn to describe the graph $G^w_f$ which establishes our
lower bound. The graph $G^w_f$ consists of three components.
The first is the modified weighted graph $G_{f}(d)$ for $d\leq \lceil(n/2c)^{1/(f+1)}\rceil$,
where $c$ is some constant to be determined later.
By Obs. \ref{obs:rel}, $n/2 \le |V(G_{f}(d))|$.
Note that $d \le (5/4)^{1/(f+1)} \cdot (n/2c)^{1/(f+1)} = (5n/8c)^{1/(f+1)}$
for sufficiently large $n$, hence $\NodesIn(f,d)=c \cdot d^{f+1} \le 5n/8$.
The second component of $G^w_f$ is a set of nodes
$X=\{x_1, \ldots, x_\chi\}$ and an additional vertex $v^{*}$ that is
connected to $u^f_{d}$ and to all the vertices of $X$.
The cardinality of $X$ is $\chi=n-\NodesIn(f,d)-1$.
The third component of $G^w_f$ is a complete bipartite graph $B$
connecting the nodes of $X$ with the leaf set $\Leaf(G_f(d))$, i.e.,
the disjoint leaf sets $\Leaf(G'_1), \ldots, \Leaf(G'_d)$.
The vertex set of the resulting graph is thus
$V=V(G_{f}(d))\cup \{v^{*}\} \cup X$ and hence $|V|=n$.
By Prop. (b) of Obs. \ref{obs:rel},
$\NLeaf(G'_i)=d^{f}=\lceil(n/2c)^{1/(f+1)}\rceil^f \ge (n/2c)^{f/(f+1)},$
hence $|E(B)| =\Theta(n \cdot d^f)$.
The following lemma follows the exact same proof as in \cite{parter2015dual}. 
\begin{lemma}\label{lem:size}[Analogue of Theorem 4.1 in \cite{parter2015dual}]
Every $f$-failure \FTBFS\ $H$ w.r.t $s=u^f_1$ in $G^w_f$ must contain all the edges of $B$.
Thus, $|E(H)|=\Omega(n \cdot d^f)$. 
\end{lemma}
We are now ready to prove the lower bound on covering value of the $(L,f)$-\RP. 
\begin{proof}[Proof of Thm. \ref{thm:lower-bound}]
Let $L=f\cdot d$ and consider the graph $G^w_f$ with the source node $s=u^f_1$.
By the construction of $G^w_f$ it holds that $(d/f)^{f+1}\leq n$. 
Let $\mathcal{G}_{L,f}$ be the optimal $(L,f)$-\RP for $G^w_f$ of minimal \CV. Our goal is to show that $|\mathcal{G}_{L,f}|=\Omega((L/f)^f)$.  We next claim that one can use this \RP (or any \RP), to compute an $f$-failure \FTBFS\ structure $H$ with $O(|\mathcal{G}_{L,f}|n)$ edges. Specifically, let $H=\bigcup_{G' \in \mathcal{G}_{L,f}} \SPT(s,G')$ where $\SPT(s,G')$ is an shortest-path tree rooted at $s$ in $G'$. It remain to show that $H$ is indeed an $f$-failure \FTBFS\ structure with respect to $s$. 

By Claim \ref{cl:hop}, every $s$-$t$ replacement path avoiding $f$ faults has $O(fd)$ edges. 
Thus, for every $P(s,t,F)$ for $|F|\leq f$ there exists a subgraph $G' \in \mathcal{G}_{L,f}$ such that $P(s,t,F)\subseteq G'$ and $F \cap G'=\emptyset$. Therefore, the $s$-$t$ path in the shortest path tree $\SPT(s,G')$ is necessarily $P(s,t,F)$. We conclude that $H \subseteq G^w_f$ is an $f$-failure \FTBFS\ structure w.r.t $s$ and with $O(|\mathcal{G}_{L,f}|n)$ edges. Combining with Lemma \ref{lem:size}, we get that $|\mathcal{G}_{L,f}|=\Omega((L/f)^f)$. 
\end{proof}

\begin{figure}[htbp]
\begin{center}
\includegraphics[width=5in]{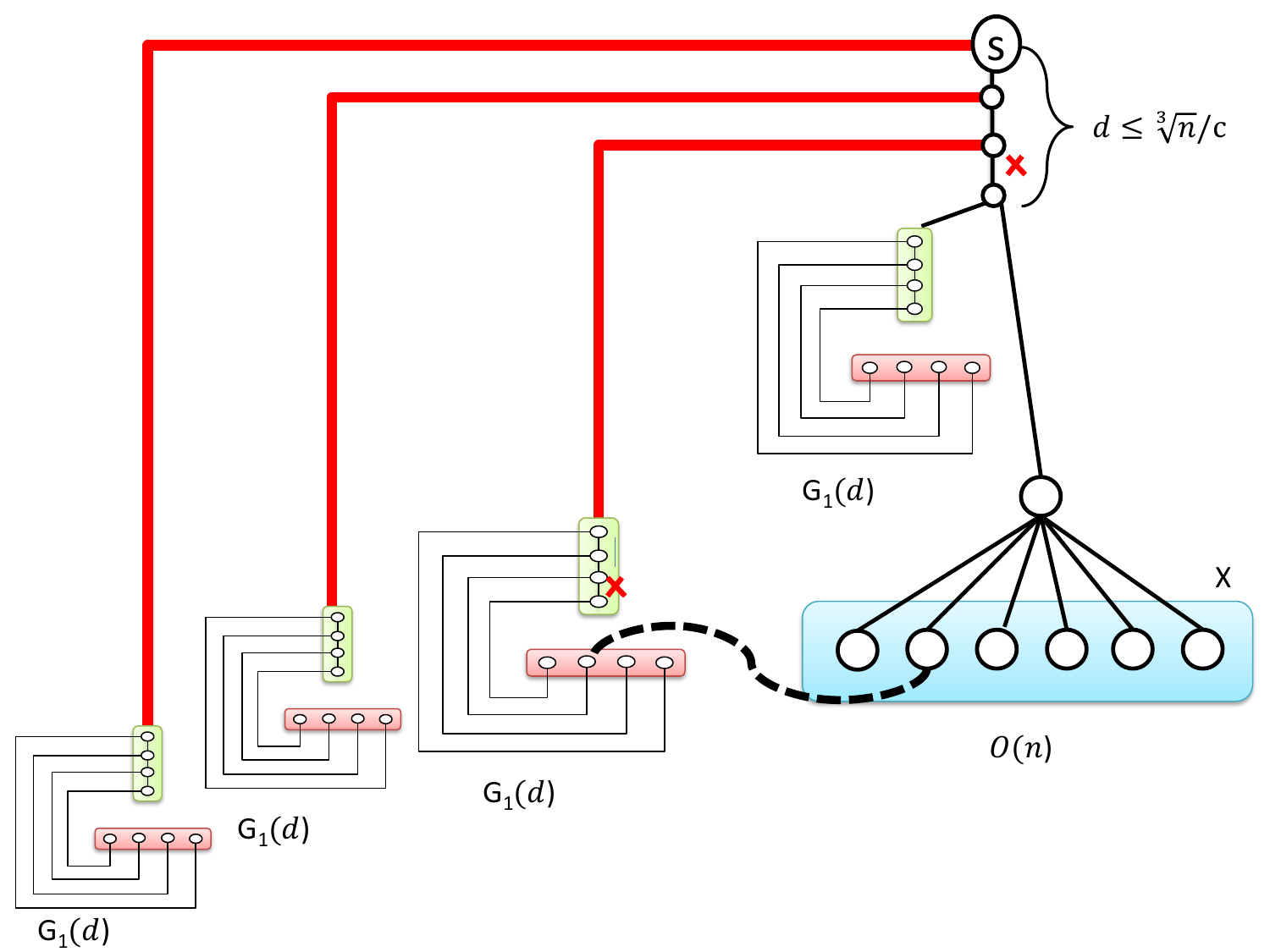}
\end{center}
\caption{Illustration of the lower-bound graph $G^w_f$ for $f=2$. The bold red edges are the only modification compared to the construction of \cite{parter2015dual}. That is, in \cite{parter2015dual} each red line correspond to a path and in our construction, it is replaced by a weighted edge whose weigh equal to the length of the path. As a result the weight of all replacement paths are preserved, but their length is edges is bounded by $O(f d)$. 
\label{fig:lowerbound}}
\end{figure}

\section{Derandomization of the Algebraic \DSO by Weimann and Yuster}
In this section, we prove Theorem \ref{thm:alg-dso-wy} by providing a derandomization of the algebraic construction of the distance sensitivity oracle of \cite{weimann2013replacement}. This construction has sub-cubic preprocessing time and sub-quadratic query time. We will use the following lemma from \cite{AlonCC19}.
\begin{lemma}\label{lem:det-hitting-set}[Lemma 2 of \cite{AlonCC19}]
Let $D_1, D_2, \ldots, D_q \subseteq V$ satisfy that $|D_i|>L$ for every $1\leq i \leq q$, and $|V|=n$. One can deterministically find in $\widetilde{O}(q \cdot L)$ time a set $R \subset V$ such that $|R|=O(n\log n/L)$ and $D_i \cap R \neq \emptyset$ for every $1\leq i \leq q$.
\end{lemma}
We start by providing a short overview of the randomized algebraic construction of \cite{weimann2013replacement}. As we will see, despite the fact that the query algorithm of \cite{weimann2013replacement} is in fact deterministic, due to the derandomization of the preprocessing part, the query algorithm will be similar to that of \cite{AlonCC19}.
Following \cite{AlonCC19}, it will be convenient to set $\epsilon=1-\alpha$. Throughout, we describe the construction for $0 < \epsilon <1$, $f=O(\log n/\log\log n)$ and a bound $L=n^{\epsilon/f}$.
We need the following definition.
\begin{definition}[Long and Short $(s,t,F)$]
A triplet $(s,t,F) \in V \times V \times E(G)^f$ is $L$-\emph{short} if $d^{L}(s,t,G\setminus F)=\dist(s,t,G\setminus F)$. That is, there \emph{exists} a $P(s,t,F)$ replacement path with at most $L$ edges in $G$. 
Otherwise,  $(s,t,F)$ is $L$-\emph{long}\footnote{In particular, for an $L$-long $(s,t,F)$ triplet it holds that \emph{every} $s$-$t$ shortest path in $G \setminus F$ has at least $L+1$ edges.}. When $L$ is clear from the context, we may omit it and write short (or long) $(s,t,F)$.
\end{definition}

\paragraph{Outline of the Weimann-Yuster \DSO.} 
The preprocessing algorithm starts by computing an $(L,f)$-\RP $\mathcal{G}_{L,f}=\{ G_1,\ldots, G_r \subseteq G\}$ for all replacement paths with at most $L$ edges, where $r=O(f n^{\epsilon} \log n)$. This \RP is generated randomly by sampling each edge in $G$ into $G_j$ independently with probability of $1-1/L$ for every $j \in \{1,\ldots,r\}$. Let $R$ be a random sample of $O(f n\log n/L)$ vertices in $G$, that we call \emph{hitting set} as they hit every replacement path segment with at least $L$ edges, w.h.p. 

Given the $(L,f)$-\RP $\mathcal{G}_{L,f}$ and the hitting set $R$, there are two variants of the algorithm. In one variant, a collection of matrices $A_1,\ldots, A_r$ is computed in in time $O(r \cdot M^{0.681}\cdot n^{2.575+\epsilon})$ for storing the all-pairs distances in $G_1,\ldots, G_r$. In an alternative variant, the algorithm computes for every subgraph $G_j \in \mathcal{G}_{L,f}$ a pair of matrices $B_j$ and $D_j$ in time $O(rMn^{2.376+\epsilon})$. The matrix $B_j$ stores the $R\times R$ distances in $G_j$ and it is computed based on a matrix $D_j$ in $O(|R|^2 n)$ time. 

For a query $(s,t,F)$, the \emph{query algorithm} first computes a collection of $O(f\log n)$ graphs $\mathcal{G}_F \subseteq \mathcal{G}_{L,f}$ that avoid all edges of $F$. For an $L$-short query, the distance $\dist_{G\setminus F}(s,t)$ is obtained by taking the minimum $s$-$t$ distance over all subgraphs $G' \in \mathcal{G}_F$.
To support $L$-long queries $(s,t,F)$, the algorithm uses the matrices $A_j$ (or the matrix pairs $D_j,B_j$) to compute a dense graph $G^F$ with vertex set $V(G^F)=R \cup \{s,t\}$. The edge weight $(x,y)$ for every $x,y \in V(G^F)$ is set to be the minimum $x$-$y$ distance over all the subgraphs in $\mathcal{G}_F$. The answer to the $(s,t,F)$ query is obtained by computing the $s$-$t$ distance in $G^F$.  In the preprocessing variant that computes the $A_j$ matrices, the query algorithm takes $\widetilde{O}(n^{2-2\epsilon/f})$ time. In the variant that computes the $B_j,D_j$ matrices, the query time is $O(n^{2-\epsilon/f})$. In the following subsections, we explain how to derandomize the preprocessing algorithm and combine it with the modified query algorithm of \cite{AlonCC19}. 

The structure of the remaining of the section is as follows. In Sec. \ref{sec:ft-trees}, we present an improved construction of a structure called Fault-Tolerant trees. Then, in Subsec. \ref{sec:det-description}, we provide a complete description of the preprocessing and query time algorithms, both will be based on the construction of the FT-trees.

\subsection{Algebraic Construction of Fault-Tolerant Trees}\label{sec:ft-trees}
For a given vertex pair $s,t$, the FT-tree $\FT_{L,f}(s,t)$ consists of $O(L^f)$ \emph{nodes}\footnote{To avoid confusion, we call the vertices of the FT-trees \emph{nodes}.}. Each node is labeled by a pair $\langle P,F \rangle$ where $P$ is an $s$-$t$ path in $G \setminus F$ with at most $L$ edges, and $F$ is a sequence of at most $f$ faults which $P$ avoids. \cite{AlonCC19} described a construction of FT-trees $\FT_{L,f}(s,t)$ for every pair $s,t$ and used it to implement the combinatorial DSO of \cite{weimann2013replacement}. The computation time of the FT-trees algorithm by \cite{AlonCC19} is $O(m\cdot n \cdot L^{f+1})$, which is too costly for our purposes (e.g., the implementation the algebraic DSO of \cite{weimann2013replacement}). 

\paragraph{Defining FT-Trees.}
Fix a pair $s,t \in V$. For every $i \in \{0,\ldots, f\}$, and every sequence of faults $F \subseteq E$, $|F|\leq f-i$, the tree $\FT_{L,i}(s,t,F)$ is defined in an inductive manner. Throughout, the paths $P^{L}(s,t,F)$ refer to \emph{some} shortest $s$-$t$ path in $G \setminus F$ with at most $L$ edges. If there are several such paths, the algorithm picks one as will be described later. 

\noindent\textbf{Base case:} The tree $\FT_{L,0}(s,t,F)$ for every $F \subseteq E$ and $|F| \leq f$ is defined as follows. 
If $d^{L}(s,t,G \setminus F)=\infty$ (i.e., there is no $s$-$t$ path with at most $L$ edges in $G \setminus F$), then  $\FT_{L,0}(s,t,F)$ is empty. Otherwise, $\FT_{L,0}(s,t,F)$ consists of a single node (root node) labeled by $\langle P^{L}(s,t,F), F \rangle$. This root node is associated with a binary search tree which stores the edges of the path $P^{L}(s,t,F)$.

\noindent\textbf{Inductive step:} Assume the construction of $\FT_{L,j}(s,t,F)$ for every $j$ up to $i$, and every $F \subseteq E$, $|F|\leq f-j$. The tree $\FT_{L,i+1}(s,t,F')$ is defined as follows for every set $F'$ of $f-(i+1)$ faults in $E$.
If $d^{L}(s,t,G \setminus F')=\infty$, then $\FT_{L,i+1}(s,t,F')$ is empty.
Assume from now on that $d^{L}(s,t,G \setminus F') <\infty$.
The root node $r$ of $\FT_{L,i+1}(s,t,F')$ is labeled by $\langle P^{L}(s,t,F'), F' \rangle$, and the edges of $P^{L}(s,t,F')$ are stored in a binary search tree. 
This root node is connected to the roots of the trees $\FT_{L,i}(s,t,F' \cup \{a_j\})$ for every 
$a_j \in P^{L}(s,t,F')$ satisfying that $d^{L}(s,t,G \setminus (F' \cup \{a_j\}))<\infty$. 
Letting, $r_j$ be the root node $\FT_{L,i}(s,t,F' \cup \{a_j\})$ (if such exists), we have:
$$\FT_{L,i+1}(s,t,F')=\{\FT_{L,i}(s,t,F' \cup \{a_j\}) \cup \{(r, r_j)\} ~\mid~ a_j \in P^{L}(s,t,F'), d^{L}(s,t,G \setminus (F' \cup \{a_j\}))<\infty\}~.$$
For $i=f$, we abbreviate $\FT_{L,f}(s,t,\emptyset)=\FT_{L,f}(s,t)$.
\begin{observation}\label{obs:nnode-ft}
Each tree $\FT_{L,f}(s,t)$ has at most $L^f$ nodes (in the case of vertex faults, it has at most $(L+1)^f$ nodes).
\end{observation}
\begin{proof}
The depth of the tree $\FT_{L,f}(s,t)$ is at most $f$. For the case of edge faults, 
each node in $\FT_{L,f}(s,t)$ has at most $L$ children as each node is labeled by a path of $\leq L$ edges. 
In the case of vertex faults, a path of at most $L$ edges, has $L+1$ vertices. 
\end{proof}

\paragraph{Algebraic Construction of FT-Trees.}
We now turn to provide a new algorithm for computing the FT-Trees $\FT_{L,f}(s,t)$ based on the $(L,f)$-\RP of Thm. \ref{thm:general-covering}. This algorithm will be applied in the preprocessing phase of the $f$-\DSO. 
The next theorem improves upon the $\widetilde{O}(m \cdot n \cdot L^{f+1})$-time algorithm provided in \cite{AlonCC19} for dense graphs. The key difference from \cite{AlonCC19} is that the algorithm of \cite{AlonCC19} is combinatorial (e.g., uses Dijkstra for shortest path computations), and our algorithm is algebraic (e.g., uses matrix multiplication). 

\begin{theorem}[Improved Computation of FT-Trees]\label{thm:ft-trees-alg}
For every $L$ and $f=O(\log n/\log\log n)$, there exists a deterministic algorithm that computes $\bigcup_{s,t \in V}\FT_{L,f}(s,t)$ in time: \vspace{-5pt}
\begin{enumerate} 
\item $\widetilde{O}((\alpha c Lf)^{f+1} \cdot L M n^{\omega})$ if $L \geq m^{1/c}$ for some constant $c$, and
\item  $\widetilde{O}((\alpha Lf \log n)^{f+1} \cdot L M n^{\omega})$ otherwise,
\end{enumerate}
\vspace{-8pt}
where $\alpha$ is the universal constant of Theorem \ref{thm:general-covering}.
\end{theorem}
The first step of the algorithm applies Theorem \ref{thm:general-covering} to compute an $(L,f)$-\RP $\mathcal{G}_{L,f}$.
Then, it applies the $APSP^{\leq L}$ algorithm of Lemma \ref{lem:shortAPSP} to compute in each $G' \in \mathcal{G}_{L,f}$, the collection of all $V(G') \times V(G')$ shortest paths $P^{L}_{G'}(s,t)$ with at most $L$ edges, for every $s,t \in V(G')$.

This computations serves the basis for the following key task in the construction of the FT-trees: Given a triplet $s,t,F$, compute $d^{L}(s,t,G \setminus F)$ and some path $P^{L}_{s,t,F}$ if such exists. 
\begin{lemma}\label{lem:tree-label-node-compute}
Consider a pre-computation of the $(L,f)$-\RP $\mathcal{G}_{L,f}$ for $f=O(\log n/\log\log n)$, and the application of algorithm $APSP^{\leq L}$ in each of the subgraphs $G' \in \mathcal{G}_{L,f}$. Then, given a triplet $(s,t,F)$, in time $\widetilde{O}(L)$, one can compute the distance $d^{L}(s,t,G \setminus F)$ and a corresponding path $P^{L}(s,t,F)$ (if such exists).
\end{lemma}
\begin{proof}
By Theorem~\ref{lem:num-collide-faults}, given the $(L,f)$-\RP $\mathcal{G}_{L,f}$, one can compute in time $\widetilde{O}(L)$ a collection of subgraphs $\mathcal{G}_F$ that fully avoid $F$. In addition, it holds that for \emph{any} $s$-$t$ path $P$ in $G \setminus F$ with at most $L$ edges, there must be exists a subgraph $G' \in \mathcal{G}_F$ that fully contain $P$. In particular, letting $P^*$ be the shortest $s$-$t$ path with at most $L$ edges in $G \setminus F$ (breaking ties in an arbitrary manner), there is a subgraph in $\mathcal{G}_F$  that fully contains $P^*$. 
Since the algorithm $APSP^{\leq L}$ is applied on each of the subgraphs $G' \in \mathcal{G}_F$, we have that 
\begin{equation}\label{eq:mindist}
d^{L}(s,t,G \setminus F)=\min_{G' \in \mathcal{G}_F}d^{L}(s,t,G')~.
\end{equation}
The desired path $P^{L}(s,t,F)$ corresponds to the output path of algorithm $APSP^{\leq L}$  in the subgraph $G'\in \mathcal{G}_F$ that minimizes the distance of Eq. (\ref{eq:mindist}).
\end{proof}
The computation of the FT-tree $\FT_{L,f}(s,t)$ for every $s,t \in V$ is described as follows. The root node is simply 
$P^L(s,t)$ as computed by applying algorithm $APSP^{L}$ in $G$. If $d^{L}(s,t,G)=\infty$, then $\FT_{L,f}(s,t)$ is empty. The computation of the binary search tree for storing $P^L(s,t)$ can be computed in $\widetilde{O}(L)$ time. Now, for every labeled node $\langle P^L(s,t,F), F \rangle$, the algorithm computes its child nodes $\langle P^L(s,t,F \cup \{a_j\}),F \cup \{a_j\} \rangle$ for every $a_j \in P^L(s,t,F)$. For that purpose, it applies the algorithm of Lemma \ref{lem:tree-label-node-compute} with input $(s,t,F \cup \{a_j\})$ for every $a_j \in P^L(s,t,F)$.  We are now ready to complete the proof of Theorem \ref{thm:ft-trees-alg}.  
\begin{proof}[Proof of Theorem \ref{thm:ft-trees-alg}]
The correctness of the algorithm follows by Lemma \ref{lem:tree-label-node-compute}. Therefore, it remains to bound the computation time. 
The computation of the $(L,f)$-\RP is done in time $O($\CV$(\mathcal{G}_{L,f})\cdot m)$. 
Applying algorithm $APSP^{\leq L}$ on every $G' \in \mathcal{G}_{L,f}$ takes $O($\CV$(\mathcal{G}_{L,f}) \cdot L M n^{\omega})$ time by Lemma \ref{lem:shortAPSP}.  The computation of each child node in the FT-tree takes $O(f L \log n)$ time, by Lemma \ref{lem:tree-label-node-compute}. By Observation \ref{obs:nnode-ft}, the total number of nodes in all the trees is bounded by $O(L^f \cdot n^2)$. Thus, the total time to compute all the FT-trees is bounded by $O($\CV$(\mathcal{G}_{L,f}) \cdot L M n^{\omega})$. The lemma holds by plugging the covering values of Theorem \ref{thm:general-covering} (the first and last bounds).
\end{proof}
The applicability of the FT-trees in the context of DSOs is expressed in the next lemma.
\begin{lemma}\label{lem:query-FT-trees}[Lemma 17 of \cite{AlonCC19}]
Given the computation of the trees $\FT_{L,f}(s,t), s,t \in V$, for every triplet $(s,t,F)$ one can compute $d^{L}(s,t,G)$ and a replacement path $P^L(s,t,F)$ (if such exists) in time $O(f^2 \log L)$.
\end{lemma}
\begin{proof}
Given $(s,t,F)$, we query the FT-tree $\FT_{L,f}(s,t)$ as follows. First check if the path $P^L(s,t)$ labeled at the root of the tree intersects $F$. If no, then output $P^L(s,t)$. Otherwise, letting $a_j \in P^L(s,t)\cap F$, we continue with the child node labeled by $P^L(s,t,\{a_j\})$. Again, if $P^L(s,t,\{a_j\}) \cap F=\emptyset$, we output that path and otherwise continues to its child node $P^L(s,t,\{a_j,a_{j'}\})$ for some $a_{j'} \in P^L(s,t,\{a_j\}) \cap F$. Using the binary search tree at each node $P^L(s,t,F')$, finding some edge $e' \in P^L(s,t,F') \cap F$ can be done in $O(f\log L)$ time. Since the depth of the tree is $f$, the total time is $O(f^2 \log L)$. 
\end{proof}

\subsection{Deterministic Preprocessing and Query Algorithms}\label{sec:det-description}
The randomized preprocessing algorithm of Weimann and Yuster has two randomized ingredients. The first is the computation of the $(L,f)$-\RP given by the subgraphs $G_1,\ldots, G_r$. The second is a computation of the set $R$ which, w.h.p., hits every $L$-length segment of every long $P(s,t,F)$ paths. Our deterministic preprocessing algorithm is presented below:

\begin{mdframed}[hidealllines=false,backgroundcolor=gray!30]
\center \textbf{Deterministic Preprocessing Algorithm}
\begin{itemize}
\item \textbf{(i): Compute FT-trees}. Using $(L,f)$-\RP of Thm. \ref{thm:general-covering}, apply Theorem \ref{thm:ft-trees-alg} to compute the collection of trees $\bigcup_{s,t} \FT_{L,f}(s,t)$ with $O(n^2\cdot L^f)$ nodes.

\item \textbf{(ii): Compute Critical Paths.} Let $\mathcal{D}_{L,f}$ be the collection of all the pairs $\langle P, F \rangle$ corresponding to the nodes of the FT-trees. Define the collection of \emph{critical paths}
$\mathcal{D}_L=\{ P ~\mid~ \langle P, F \rangle \in \mathcal{D}_{L,f}, |P|\in [L/4, L]\}$ which consists of all sufficiently long paths. 

\item \textbf{(iii): Compute Hitting Set for the Critical Paths.} Apply the algorithm of Lemma \ref{lem:det-hitting-set} to compute  a hitting set $R \subseteq V$ for the paths in $\mathcal{D}_L$ where $|R|=O(n\log n/L)$. 
\end{itemize}
\end{mdframed}
This completes the description of the preprocessing algorithm. We note that the computation of the FT-trees substitutes the $A_j,B_j,D_j$ matrices used in \cite{weimann2013replacement}.


\begin{lemma}[Preprocessing time]\label{lem:wy-pre}
The preprocessing time of the deterministic algorithm is bounded by 
\begin{enumerate} 
\item $\widetilde{O}((\alpha c Lf)^{f+1} \cdot L M n^{\omega})$ if $L \geq m^{1/c}$ for some constant $c$,
\item  $\widetilde{O}((\alpha Lf \log n)^{f+1} \cdot L M n^{\omega})$ otherwise,
where $\alpha$ is the universal constant of Theorem \ref{thm:general-covering}.
\end{enumerate}
\end{lemma}
\begin{proof}
The computation time is dominated by the computation of the FT-trees, see Theorem \ref{thm:ft-trees-alg}.
The FT-trees consists of $O(n^2 \cdot L^f)=O(n^{2+\epsilon})$ labeled nodes, and thus $|\mathcal{D}_L|=O(n^{2+\epsilon})$. By Lemma \ref{lem:det-hitting-set}, the computation of the hitting set $R$ takes $O(n^{2+\epsilon+\epsilon/f})$ time, and $|R|=O(n\log n/L)$. 
\end{proof}
By setting the matrix multiplication exponent to $\omega=2.373$, and $\epsilon=1-\alpha$, Lemma \ref{lem:wy-pre} achieves the bound of Theorem \ref{thm:alg-dso-wy}.

\paragraph{The Query Algorithm.} Once the FT-trees are computed, the query algorithm is the same as in
\cite{AlonCC19}, for completeness we describe it here. Note that in contrast to \cite{AlonCC19}, we do not assume here that the shortest path ties are decided in a consistent manner. Thus the correctness of the procedure is somewhat more delicate. Given a short query $(s,t,F)$, i.e., $d^{L}(s,t,G\setminus F)=\dist(s,t,G\setminus F)$, the desired distance $d^{L}(s,t,G\setminus F)$ can be computed in time $O(f^2\log L)$ by using the query algorithm of Lemma \ref{lem:query-FT-trees}. From now on assume that the query $(s,t,F)$ is long. Unlike \cite{weimann2013replacement} we would not be able to show that there are few subgraphs in the $(L,f)$-\RP $\mathcal{G}_{L,f}$ that fully avoid\footnote{There are $\widetilde{O}(L)$ such subgraphs which is too costly for our purposes.} $F$. Nevertheless, we will still be able to efficiently compute the dense graph $G^F$, e.g., within nearly the same time bounds as in \cite{weimann2013replacement}. 
Recall that $R$ is the hitting-set of the critical set of replacement paths. The vertex set of the graph $G^F$ is given by $V^F=R \cup \{s,t\}$, and the weight of each edge $(x,y) \in V^F \times V^F$ is given by $w(x,y)=d^{L}(x,y,G\setminus F)$. This weight can be computed by applying the query algorithm of Lemma \ref{lem:query-FT-trees} on the FT-tree $\FT_{L,f}(x,y)$ with the query $(x,y,F)$. 

To answer the $(s,t,F)$ query it remains to compute the $s$-$t$ distance in the dense graph $G^F$. Using the method of feasible price functions and in the exact same manner as in \cite{weimann2013replacement}, this computation is done in $\widetilde{O}(|E(G^F)|)=\widetilde{O}(n^{2-2\epsilon/f})$. This completes the description of the query algorithm. Given the computation of the FT-trees in the preprocessing step, by Lemma \ref{lem:query-FT-trees} the computation of the graph $G^F$ takes $O(|E(G^F)|\cdot f^2 \log L)=\widetilde{O}(n^{2-2\epsilon/f})$ time. This matches the query time of Weimann and Yuster \cite{weimann2013replacement} (up to poly-logarithmic terms). We finalize the section by showing the correctness of the query algorithm. Due to the fact that we do not assume uniqueness of shortest paths as in \cite{AlonCC19}, the argument is more delicate.
\begin{claim}\label{cl:wy-query-corr}
$\dist(s,t,G^F)=\dist(s,t,G \setminus F)$.
\end{claim}
\begin{proof}
The correctness for the short queries $(s,t,F)$ follows by the correctness of Lemma \ref{lem:query-FT-trees}.
Consider a long query $(s,t,F)$ and let $P(s,t,F)$ be the $s$-$t$ shortest path in $G \setminus F$ with the minimal number of edges. If there are several such paths, pick one in an arbitrary manner. 
By definition, $P=P(s,t,F)$ has at least $L$ edges. 
Partition it into segments of length\footnote{E.g., partition $P(s,t,F)$ into consecutive segments of length $L/4$, while the last segment have length at most $L/2$.} $[L/4,L/2]$ and let $s_i$-$t_i$ be the endpoints of the $i$th segment. That is, $P=P[s_1=s,t_1=s_2] \circ P[s_2,t_2] \circ \ldots P[s_\ell,t_\ell=t]$. 

By the definition of $P$, \emph{every} $s_i$-$t_i$ shortest path in $G \setminus F$ \emph{must} have at least $L/4$ edges. To see this, assume towards contradiction otherwise that there exists a pair $s_i,t_i$ with a shorter (in number of edges) $s_i$-$t_i$ shortest path in $G \setminus F$. This implies that we can obtain an $s$-$t$ shortest path $P''$ of the same weight but with fewer edges, contradiction to the minimality (in edges) of $P$. 
Since $d^{L/2}(s_i,t_i,G \setminus F)=\dist(s_i,t_i,G \setminus F)$ for every $i \in \{1,\ldots, \ell\}$, there is an $s_i$-$t_i$ path $P^{L}(s_i,t_i,F)$ of length at most $L$ in the FT-tree $\FT_{L,f}(s_i,t_i)$. Specifically, this path can be found by applying the query algorithm of Lemma \ref{lem:query-FT-trees} with the query $(s_i,t_i,F)$. By Lemma \ref{lem:query-FT-trees}, this results in the distance
$d^{L}(s_i,t_i,G \setminus F)$ along with a path $P^{L}(s_i,t_i,F)$. 

Consider now an alternative $s$-$t$ path $P'=P^{L}(s_1,t_1,F) \circ P^{L}(s_2,t_2,F)\circ \ldots \circ P^{L}(s_\ell,t_\ell,F)$. Since $d^{L/2}(s_i,t_i,G \setminus F)=\dist(s_i,t_i,G \setminus F)$ for every $i \in \{1,\ldots, \ell\}$, we have that $P' \cap F=\emptyset$ and $\len(P')=\len(P)=\dist(s,t,G\setminus F)$. 

By definition, every $P^{L}(s_i,t_i,F) \in \mathcal{D}_{L,f}$, and since $P^{L}(s_i,t_i,F)$ has at least $L/4$ edges and at most $L$ edges, $P^{L}(s_i,t_i,F) \in \mathcal{D}_L$.
Since $R$ is a hitting-set of all paths in $\mathcal{D}_L$, there exists some $x_i \in P^{L}(s_i,t_i,F) \cap R$ for every $i$.
This implies that $P'$ can be written as a concatenation of replacement path segments each with at most $L$ edges and with both endpoints in $V(G^F)=R \cup \{s,t\}$. Let $\{s=x_0, x_1,\ldots, x_k, x_{k+1}=t\}$ be the ordered set of the representatives of the $V(G^F)$ vertices on $P'$. By the description of the query algorithm, for every $i \in \{0,\ldots, k\}$, it holds that $w(x_i,x_{i+1})=d^{L}(x_i,x_{i+1}, G \setminus F)$. By the above argument, $d^{L}(x_i,x_{i+1}, G \setminus F)=\dist(x_i,x_{i+1},G \setminus F)$. In addition, for \emph{every} pair $x,y \in V(G^F)$, $w(x,y)=d^{L}(x,y, G \setminus F)\geq \dist(x,y,G\setminus F)$. We therefore conclude that $\dist(s,t,G^F)=\len(P')=\dist(s,t,G \setminus F)$.
\end{proof}

\section{Derandomization of Fault Tolerant Spanners}
We next consider the applications of the $(L,f)$-\RP to deterministic constructions of fault-tolerant spanners resilient to at most $f$ \emph{vertex} faults. For a given $n$-vertex (possibly) weighted graph $G=(V,E)$, a subgraph $H\subseteq G$ is an $f$-fault tolerant $(\alpha,\beta)$-spanner if
$$\dist(s,t,H \setminus F)\leq \alpha \cdot \dist(s,t,G \setminus F)+\beta, \mbox{~for every~} s,t \in V, F \subseteq V, |F|\leq f~.$$
When $\beta=0$, the spanner is called multiplicative spanner, denoted by $f$-fault tolerant $t$-spanner for short, $t$ is the stretch factor. When $\alpha=1$, the spanner is additive. 

\subsection{Multiplicative Vertex Fault-Tolerant Spanners}
Chechik, Langberg, Peleg, and Roddity \cite{chechik2010fault} presented the first non-trivial construction of $f$ fault-tolerant multiplicative spanners resilient to vertex faults. The size overhead of their construction (compared to standard spanner) is $k^f$, that is, exponential in the number of faults. Dinitz and Krauthgamer \cite{dinitz2011fault} provided a simpler and sparser solution by using the notion of \RP{}s. They showed:
\begin{theorem}[Theorem 1.1 of \cite{dinitz2011fault}]\label{thm:mike}
For every graph $G=(V,E)$ with positive edge lengths and odd $t\geq 3$, 
there is an $f$-fault tolerant $t$-spanner with size $O(f^{2-2/(t+1)}\cdot n^{1+2/(t+1)}\log n)$. 
\end{theorem}
This theorem is a consequence of a general conversion scheme that turns any $\tau(n,m)$-time algorithm for constructing $t$-spanners with size $s(n)$ into an 
algorithm for constructing $f$-fault tolerant $t$-spanner with size $O(f^3\log n \cdot s(2n/f))$ and time complexity $O(f^3\log n \cdot \tau(2n/f,m))$. Specifically, applying this conversion to the greedy spanner algorithm yields an $f$-fault tolerant $(2k-1)$-spanner with $O(f^3\log n\cdot (n/f)^{1+1/k})$ edges in time
$O(f^3\log n k \cdot m \cdot (2n/f)^{1+1/k})$. 
In this section we provide the derandomization of Theorem 2.1 of \cite{dinitz2011fault} (which used to obtain Theorem 1.1) and show: 
\begin{theorem}[Derandomized of Theorem 2.1 of \cite{dinitz2011fault}]\label{thm:spanner-mult-full}
If there is a deterministic algorithm $\cA$ that on every $n$-vertex $m$-edge graph builds a $t$-spanner of size $s(n)$ and time $\tau(n,m,t)$, then there is an algorithm that on any such graph builds an $f$-fault tolerant $t$-spanner of:
\begin{enumerate}
\item size $O(f^3 \cdot s(n/f))$ and time $O(f^3 (\tau(n/f,m,t)+ m))$, if $f \geq n^{1/c}$ for some constant $c\in \mathbb{N}$.
\item size $O(\log^5 n \cdot s(n/f))$ and time $O(\log^5 n (\tau(n/f,m,t)+ m))$, if $f \leq \log n$
\item size $O((f\log n)^3 \cdot s(n/f))$ and time $O((f\log n)^3 (\tau(n/f,m,t)+ m))$, if $f \in [\log n, n^{o(1)}]$.
\end{enumerate}
\end{theorem}
\begin{proof}
The algorithm applies the vertex variant of Theorem \ref{lem:num-collide-paths} to compute $(L=2, f)$ \RP $\mathcal{G}$. Then, it applies the fault-free algorithm $\mathcal{A}$ for computing the $t$-spanner $H_j$ for each subgraph $G_j \in \mathcal{G}$. The output spanner $H=\bigcup_{j=1}^{r} H_j$ is simply the union of all these spanner subgraphs.

We first consider correctness. Fix a replacement-path $P(s,t,F)$. It is required to show that $\dist(s,t,H \setminus F)\leq t \cdot \dist(s,t,G \setminus F)$ and thus it is sufficient to show that 
$\dist(u,v,H \setminus F)\leq w(u,v)$ for every edge $(u,v)\in P(s,t,F)$, where $w(u,v)$ is the weight of the edge $(u,v)$ in $G$. Since $\mathcal{G}$ is an $(2,f)$-\RP, there exists a subgraph $G_j \in \mathcal{G}$ satisfying that $(u,v) \in G_j$ and $F \cap V(G_j)=\emptyset$. Thus, the $t$-spanner $H_j \subseteq H$ satisfies that $\dist(u,v,H_j \setminus F)=\dist(u,v,H_j)\leq t w(u,v)$, as desired. 

We now turn to show that the computation time is $O(|\mathcal{G}|\cdot (\tau(n/f,m,t)+m))$ and that the size of the spanner is $O(|\mathcal{G}| \cdot  s(n/f,m,t))$.  By Theorem \ref{lem:num-collide-paths}(I5v), 
we get that $|V(G_j)|=O(n/f)$ for every $G_j \in \mathcal{G}$. The bounds then follows by plugging the covering value $|\mathcal{G}|$ and the computation time of the covering of Theorem \ref{thm:general-covering}.
\end{proof}

\subsection{Nearly Additive Fault-Tolerant Spanners} In \cite{BraunschvigCPS15}, the approach of \cite{dinitz2011fault} was extended to provide vertex fault-tolerant spanners with nearly additive stretch. 

\begin{theorem}\label{thm:additiveft-spanner}[Derandomization of Theorem 3.1 of \cite{BraunschvigCPS15}]
Let $\mathcal{A}$ be an algorithm for computing $(\mu,\alpha)$-spanner of size $O(n^{1+\delta})$ in time $\tau$ for an $n$-vertex $m$-edge graph $G=(V,E)$. Set  $L=\lceil \alpha \cdot \epsilon^{-1}\rceil+1$.  Then, for any $\epsilon>0$ and $f \leq L$, one can compute an $f$-vertex fault-tolerant $(\mu+\epsilon,\alpha)$-spanner with: 
\begin{enumerate}
\item $O((c'f L)^{f+1} \cdot n^{1+\delta})$ edges in time $\widetilde{O}((f c'L)^{f+1}\cdot \tau)$, if $L\geq n^{1/c}$ for some constant $c \in \mathbb{N}$. 

\item $O((c'f L)^{f+2} \cdot \log n\cdot n^{1+\delta})$ edges in time $\widetilde{O}((c'f L)^{f+2} \cdot \log n \cdot \tau)$, if $L \leq \log n$.

\item $O((c'f L \log n)^{f+1} \cdot n^{1+\delta})$ edges in time $\widetilde{O}((c'f L \log n)^{f+1}\cdot \tau)$, otherwise,
\end{enumerate}
for some constant $c'$. 
\end{theorem}
\begin{proof}
The proof follows the exact same line as Theorem 3.1 of \cite{BraunschvigCPS15} only when using 
Theorem \ref{thm:general-covering} to build an $(L+1,f)$-\RP $\mathcal{G}=\{G_1,\ldots, G_\gamma\}$.
It then applies algorithm $\mathcal{A}$ on each of these subgraphs, and take the union of the output spanner as the final subgraph $H$. The size and time bounds are immediate by Theorem \ref{thm:general-covering}. 
To see the stretch argument, it is sufficient to show that for any path of length at most $L$ in $G \setminus F$, there is a corresponding path in $H \setminus F$ of bounded length. The stretch argument for longer paths is obtained by decomposing it into $L$-length segments (except perhaps for the last segment), and accumulating the additive stretch from each segment. 
Fix an $L$-length path $P \subseteq P(s,t,F)$, and let $u,v$ be the endpoints of $P$.
Since $\mathcal{G}$ is an $(L+1,f)$-\RP, w.h.p., there exists a subgraph $G_i \in \mathcal{G}$ such that $P \subseteq G_i$ and $F\cap G_i=\emptyset$. 
Since $H_i$ is an $(\mu, \alpha)$-spanner for $G_i$, we have that 
$$\dist(u,v,H_i \setminus F)=\dist(u,v,H_i)\leq \mu\cdot L+\alpha~.$$
Partition any path $P(s,t,F)$ into $\lceil (1/L)\cdot \dist(s,t,G\setminus F)\rceil$ segments each of length at most $L$. We then have that
$$\dist(s,t,H \setminus F)\leq \mu \cdot \dist(s,t,G \setminus F)+\alpha \cdot \lceil (1/L)\cdot \dist(s,t,G \setminus F)\rceil~.$$
Since $1/L<\epsilon/\alpha$, the stretch bound holds.
\end{proof}

\subsection*{Acknowledgment}
We would like to thank Swastik Kopparty, Gil Cohen, and Amnon Ta-Shma for discussion on coding theory, Moni Naor for discussion on universal hash functions, and Eylon Yogev for various discussions.
 
\bibliographystyle{alpha} 
\bibliography{crypto}

\appendix

\section{Comparison with \cite{parter2019small} and \cite{BodwinDinitz20}}\label{sec:comparison}
In \cite{parter2019small}, the second author provided the first deterministic constructions of $(L,f)$-\RP for $L\geq f$. The notion of $(L,f)$-\RP is introduced for the first time in the current paper, and in \cite{parter2019small} the construction is referred to as a \emph{derandomization of the FT-sampling technique}. The construction of \cite{parter2019small,parter2019small-arxiv} is based on a computation of a family of perfect hash functions $\mathcal{H}=\{h: [n] \to [2(L+f)^2]\}$ with $\poly(Lf\log n)$ functions. The covering subgraph family $\mathcal{G}$ of \cite{parter2019small,parter2019small-arxiv} consists of $|\mathcal{H}|\cdot (4Lf)^{2f}=(4Lf\log n)^{O(1)+2f}$ subgraphs. In the context of \cite{parter2019small}, it was sufficient for the value of the covering to be polynomial in $L$, and for the computation time to be polynomial in $n$. Also note that despite the fact that \cite{parter2019small,parter2019small-arxiv} explicitly considers the setting where $L \geq f$, their construction can be extended to provide a covering of value $\poly(f\log n)$ also for the case\footnote{This is similarly to the random construction of $(L,f)$-RPC, where the sampling probability also differs between when $L \leq f$ and $L > f$.} of $L\leq f$. Specifically, this can be done by 
applying very minor modifications to Lemma 17 of \cite{parter2019small-arxiv}: set $a=f$ and $b=L$, then 
let the set $S_{h,i_1,i_2,\ldots, i_b}$ of the lemma be given by
\begin{equation}\label{eq:fault-sets}
S_{h,i_1,i_2,\ldots, i_b}=\{ \ell \in [n] ~\mid~ h(\ell) \in \{i_1,i_2,\ldots, i_b\}\}, \forall h \in \mathcal{H} \mbox{~and~} i_1,i_2,\ldots, i_b\in [2(L+f)^2]~.
\end{equation}
I.e., the only modification for $L\leq f$ is in replacing the $\notin$ sign with $\in$ in Eq. (\ref{eq:fault-sets}). The argument then follows in a symmetric manner as in the proof of Lemma 17 of \cite{parter2019small-arxiv}. To summarize, the construction of \cite{parter2019small,parter2019small-arxiv} provides an $(L,f)$-\RP of value $\poly(\min\{L,f\}\log n)$.

In this work, we considerably optimize the construction of \cite{parter2019small} in several ways. First, we almost match the optimal values $(L,f)$-\RP{}s for a wide range of parameters (e.g., when $f=O(1)$), providing a polynomial improvement in $\max\{L,f\}$ compared to \cite{parter2019small,parter2019small-arxiv}. 
Second, we establish several key properties of $(L,f)$-\RP{}s (e.g., Theorems \ref{lem:num-collide-paths}) which have extensive applications. Those properties follow immediately by the randomized construction, and are proven in a quite natural manner in our deterministic setting as well. For example, in order to provide a ``perfect" derandomization of Weimann and Yuster DSO \cite{weimann2013replacement} as provided in the paper, we must use our nearly optimal constructions of $(L,f)$-\RP{}s. Using the suboptimal $(L,f)$-RPC constructions of \cite{parter2019small,parter2019small-arxiv} lead to a polynomially larger query time compared to that of \cite{weimann2013replacement}. Third, we provide the first lower bound for the values of the $(L,f)$ covering. We also note that our techniques differ from \cite{parter2019small,parter2019small-arxiv} and are based on various coding schemes.

Independent to our work, very recently \cite{BodwinDinitz20}  presented a (randomized) slack version of the greedy algorithm to obtain (vertex) fault-tolerant spanners of \emph{optimal} size. To derandomize their construction, \cite{BodwinDinitz20} provided a deterministic construction $(L=2,f)$-\RP (using our terminology) with additional properties. The work of \cite{BodwinDinitz20} leaves a gap in the running time depending on the value of the number of faults, $f$. Specifically, for $f\geq n^c$ for some constant $c$, their derandomization matches the bounds of their randomized construction. In contrast, for smaller values of $f$, there is a gap of $\poly(f)$ factor in the running time. In our work, using the generalized construction of $(L,f)$-\RP with $L=2$ and in particular using  Theorem \ref{lem:num-collide-paths} (instead of Theorem 5.3 of \cite{BodwinDinitz20}) we close this gap. 

Elaborating, their non-optimality in derandomization stems from a not completely tight analysis of some additional properties of the \RP that they construct, and in order to compensate for this analysis, they rely on using ``bulkier'' objects such as  almost $k$-wise independent families in a black-box manner.

More formally, by using Theorem \ref{lem:num-collide-paths} instead of Lemma 5.3 of \cite{BodwinDinitz20}, we show:
\begin{lemma}[Improvement of Thm. 2.1 of \cite{BodwinDinitz20}]\label{lem:imp}
There is a deterministic algorithm which computes an $f$-(vertex) fault tolerant $(2k-1)$ spanner with at most $O(f^{1-1/k}n^{1+1/k})$ edges in time $\widetilde{O}(f^{1-1/k}n^{2+1/k}+m\cdot f^2)$ (matching the bounds of the randomized construction of Theorem 1.1 of \cite{BodwinDinitz20}).
\end{lemma}
\begin{proof}
Let $\mathcal{G}_{2,f}$ be the $(2,f)$-RPC of Theorem \ref{lem:num-collide-paths}. By claim (I2) of Theorem \ref{lem:num-collide-paths}, there is a collection of $\widetilde{O}(f)$ subgraphs $\mathcal{G}_{P=e}$ that contain both endpoints of $e$. This is the analogue to the set $L_e$ defined by \cite{BodwinDinitz20}.
For every fixed set $F$ of at most $f$ vertex faults, let $\mathcal{G}_{e,F}$ be the subset of subgraphs in $\mathcal{G}_e$ that fully avoid $F$. To provide a spanner of optimal size in Alg. 2 of \cite{BodwinDinitz20}, it is required that for every $P,F$, the ratio $|\mathcal{G}_{e,F}|/|\mathcal{G}_{e}|\geq c$, for some constant $c$. 
Indeed, by claim (I4) it holds that $|\mathcal{G}_{P,F}| \geq |\mathcal{G}_P|/2$ for every $F$. 
By setting $\tau$ to $1/3$ in Alg. 2 of \cite{BodwinDinitz20} the correctness and the size of the spanner follows by Lemma 5.4 and 5.5 in \cite{BodwinDinitz20}.  (In contrast, in \cite{BodwinDinitz20} the ratio $|\mathcal{G}_{e,F}|/|\mathcal{G}_e|$ depends also on some parameter $\delta$ of their universal hash function). 

It remains to bound the running time. By \ref{lem:num-collide-paths}, for every $e$, computing the collection $\mathcal{G}_e$ takes $\widetilde{O}(f^2)$ time. Thus, taking $\widetilde{O}(f^2 m)$ time for all the edges.
Next, for every fixed edge $e$, computing the vertices of each subgraph in $\mathcal{G}_e$ takes $\widetilde{O}(n/f)$ time per subgraph and $|\mathcal{G}_e|\cdot\widetilde{O}(n/f)=\widetilde{O}(fn)$ in total. 
The rest of the time argument works line by line as in Lemma 5.6 of \cite{BodwinDinitz20}.
\end{proof}

\section{Missing Proofs}\label{sec:miss-proof}
\APPENDRANDLFRPC

\section{Improved \RP given Input Sets} \label{sec:input-RPC}
In  this section, we show an improved \RP computation based on a given input set $\mathcal{D}$. 
Specifically, we consider a relaxed notion of the problem as suggested by Alon, Chechik, and Cohen \cite{AlonCC19} and provide an $(L,f)$-\RP for this relaxed notion with nearly optimal covering value.  The main result of this section is the following.

\begin{theorem}\label{thm:relaxed-RPC}
Let $L,f$ be integer parameters such that $L\ge f$. 
There exists an algorithm $\mathcal{A}$ that takes as input a graph $G$ on $n$ vertices and $m$ edges and  a list $\mathcal{D}=\{ ( P_1, F_1 ), \ldots, ( P_k, F_k )\}$ of $k$ pairs of $L$-length replacement paths $P_i$ and set of faults $F_i$ that it avoids\footnote{In the problem statement of \cite{AlonCC19}, $k=O(n^{2+\epsilon})$.} and outputs a restricted $(L,f)$-\RP $\mathcal{G}(\mathcal{D})$ satisfying that for every $(P_i, F_i ) \in \mathcal{D}$, there is a subgraph $G' \in \mathcal{G}(\mathcal{D})$ that contains $P_i$ and avoids $F_i$.  
Moreover, the running time of $\mathcal{A}$ is 
     $(m+k)\cdot (\log m)^{O(1)}\cdot    (\alpha Lf\log m)^{f}$,
where  $\alpha\in\mathbb N$ is some small universal constant. 
\end{theorem}

Towards the goal of proving Theorem~\ref{thm:relaxed-RPC}, we start by showing that for every $a,b,N$, given an explicit set $\mathcal{S}=\{( A,B ) ~\mid~ A, B\subseteq [N], |A| \leq a,|B| \leq b, A \cap B =\emptyset\}$, there exists considerably smaller set of hash function $\mathcal{H}_{\mathcal{S}}=\{h: [N] \to [q]\}$ with the following property. For every $(A,B)\in \mathcal{S}$, there exists a function $h \in \mathcal{H}_{\mathcal{S}}$ that does not collide on $(A,B )$. The next lemma should be compared with Corollary~\ref{cor:RS}. The latter works for any pair of disjoint sets $A, B$, while the next lemma satisfies the collision-free property for every $( A,B )\in \mathcal{S}$. This allows us to obtain a considerably smaller family of functions.

\begin{lemma}\label{cl:small-perfectexp}
Let $b\le a\le N$ all be integers. There is an algorithm $\mathcal{A}$ which given a set $\mathcal{S}=\{ ( A,B) \rangle ~\mid~ A, B\subseteq [N], |A|\leq a,|B|\leq b, A \cap B =\emptyset\}$ and a \code{N,a,b,\ell}{q}{}-Strong \HM $\H$ as input, and outputs a collection of hash functions $\mathcal{H}_{\mathcal{S}}=\{h: [N] \to [q]\}$ such that the following holds:
\begin{itemize}
\item{(P1)} For every $(A,B) \in \mathcal{S}$,  $\exists h \in \mathcal{H}_{\mathcal{S}}$ such that  $\forall (x,y)\in A\times B$, we have $h(x)\neq h(y)$.
\item{(P2)} $|\mathcal{H}_{\mathcal{S}}|=O(\log |\mathcal{S}|)$.
\end{itemize}
Moreover, $\mathcal{A}$ runs in time ${O}(T_{\H}+a\cdot \ell\cdot |\mathcal{S}|)$, where $T_\H$ is the computation time of $\H$.
\end{lemma}
\begin{proof}

For every  $(A,B)\in \mathcal{S}$, let $\mathcal{H}_{A,B}=\{i\in[\ell]\mid \forall (x,y)\in A\times B,\ h_i(x)\neq h_i(y)\}$. Since $\H$ is a Strong \HM, we have $|\mathcal{H}_{A,B}|\geq \nicefrac{\ell}{2}$. The desired collection of hash functions  $\mathcal{H}_{\mathcal{S}}$ is obtained by computing a small hitting set for the sets $\{\mathcal{H}_{A,B} ~\mid~ ( A,B ) \in \mathcal{S}\}$. This can be done by the algorithm of Lemma \ref{lem:det-hitting-set}.
 
We next analyze the computation time. First we   compute the $\ell\times N$ Boolean matrix  $M_{\H}$ corresponding to $\H$ where the $(i,x)^{\text{th}}$ entry of $M_{\H}$ is simply $h_{i}(x)$.  After the computation of $M_{\H}$ we simply go over each $(A,B)\in\mathcal{S}$ and compute the sets $\mathcal{H}_{A,B}$. The computation time of all the $\mathcal{H}_{A,B}$ sets takes $O(|\mathcal{S}|\cdot \ell\cdot (a+b))$ time. Then, the set $\mathcal{H}_{\mathcal{S}}$ is computed by applying the hitting set algorithm of Lemma \ref{lem:det-hitting-set} with parameters $n=\ell, L=\ell/2$, and $q=|\mathcal{S}|$. Thus the total computation time is $O(T_{\H}+a\cdot \ell\cdot |\mathcal{S}|)$. 
\end{proof}

Finally, we show how to compute a covering graph family $\mathcal{G}^*_{L,f}$ for the critical set $\mathcal{D}_{L,f}$.  The proof of the following lemma is similar to that of Theorem \ref{thm:general-covering}, but it is based on Lemma~\ref{lem:strongcodes} and Lemma \ref{cl:small-perfectexp} rather than on Theorem~\ref{thm:HM}.
For the sake of brevity, we only prove the below for Reed-Solomon codes, as it suffices to give the claim in Theorem~\ref{thm:relaxed-RPC}.

\begin{lemma}\label{lem:smaller-covering-family}
Given a critical set $\mathcal{D}$, there is a deterministic algorithm for computing an $(L,f)$-\RP $\mathcal{G}(\mathcal{D})$ of cardinality $O((2Lf\log N)^{f}\cdot \log (|\mathcal{D}|))$ in time $\widetilde{O}((2Lf\log N)^{f+1}\cdot m+(n \cdot L^f)\cdot (L\cdot f)^2)$.
\end{lemma}
\begin{proof}
Set $a=L$, $b=f$ and $N=m$ and let $\mathcal{S}=\mathcal{D}_{L,f}$. Note that since each pair in $\mathcal{D}$ is given by $(P, F )$ where $P \cap F=\emptyset$, $|P|\leq L$ and $|F|=
\leq f$, the set $\mathcal{S}$ is a legal input to Claim \ref{cl:small-perfectexp} combined with Reed-Solomon Strong \HM from Lemma~\ref{lem:strongcodes}. 

We then safely apply Claim \ref{cl:small-perfectexp} to compute a collection of hash functions $\mathcal{H}_{\mathcal{S}}=\{h: [N] \to [2ab\log N]\}$ that satisfies properties (P1) and (P2).
For every $h \in \mathcal{H}_{\mathcal{S}}$ and for every subset $i_1,\ldots, i_b \in [1,2ab\log N]$, define:
\begin{equation}\label{eq:graph-def-hash-smaller}
G_{h,i_1,i_2,\ldots, i_b}=\{ e_\ell \in E(G) ~\mid~ h(\ell) \notin \{i_1,i_2,\ldots, i_b\}\}~.
\end{equation}
Overall, $\mathcal{G}(\mathcal{D})=\{G_{h,i_1,i_2,\ldots, i_b} ~\mid~ h \in \mathcal{H}_{\mathcal{S}}, i_1,i_2,\ldots, i_b \in [1,2ab\log N]\}$. 

The cardinality of $\mathcal{G}^w_{L,f}$ is bounded by $O(|\mathcal{H}_{\mathcal{S}}|\cdot (2Lf\log N)^{b})=O((2Lf\log N)^{f}\cdot \log(|\mathcal{D}|))$,. To show that $\mathcal{G}(\mathcal{D})$ satisfies properties of Theorem~\ref{thm:relaxed-RPC}, it is sufficient to show that it resiliently covers all the pairs in the critical set $\mathcal{D}_{L,f}$. 
Fix $(P,F)\in \mathcal{D}$ where $P$ is a $u$-$v$ path. We will show that there exists at least one subgraph $G' \in \mathcal{G}(\mathcal{D})$ satisfying that $P \subseteq G'$ and $F \cap G'=\emptyset$. 
Letting $A=E(P)$ and $B=F$, we have that $(A,B) \in \mathcal{S}$.
By property (P1) of $\mathcal{H}_{\mathcal{S}}$, there exists a function $h$ that does not collide on $A,B$. That is, there exists a function $h \in \mathcal{H}$ such that $h(i) \neq h(j)$ for every $i \in A$ and $j \in B$.
Thus, letting $B=\{s_1,\ldots, s_b\}$ and $i_1=h(s_1),\ldots, i_b=h(s_b)$, we have that $h(s'_j)\notin \{i_1,\ldots, i_b\}$ for every $s'_j \in A$. Therefore, the subgraph $G_{h,i_1,i_2,\ldots, i_b}$ satisfies that 
$A  \subseteq S_{h,i_1,i_2,\ldots, i_b}$ and $B \cap  S_{h,i_1,i_2,\ldots, i_b}=\emptyset$. 

Finally, we analyze the computation time. By Cl. \ref{cl:small-perfectexp}, the computation of $\mathcal{H}_{\mathcal{S}}$ takes $\widetilde{O}(L f \cdot m+(n \cdot L^f)\cdot (L\cdot f)^2)$ time.
Next, consider the evaluation all functions in $\mathcal{H}_{\mathcal{S}}$ on all the elements in $[m]$. This takes $\widetilde{O}(\log(|\mathcal{D}|)\cdot m)=\widetilde{O}(m \cdot \log(|\mathcal{D}|))$. Next, for a fixed hash function $h \in \mathcal{H}_{\mathcal{S}}$ and $i_1,i_2,\ldots, i_b \in [1, 2ab\log N]$, the computation of the subgraph $G_{h,i_1,i_2,\ldots, i_b}$ can be done in $O(m)$ time. Thus, the computation of all the subgraphs takes $\widetilde{O}((L\cdot f\log N)^{f}\cdot \log(n\cdot L^f)\cdot m)$ time.
\end{proof}
Finally, we show that for every $(P,F) \in \mathcal{D}$, there are at most $O(\log N)$ subgraphs in $\mathcal{G}(\mathcal{D})$ that contain no edge from $F$.
\begin{lemma}\label{lem:nice-num-collide-faults}
Fix $(P,F) \in \mathcal{D}$. Then, $|\{G' \in \mathcal{G}^w_{L,f} ~\mid~ F \cap G'=\emptyset\}|=O(\log N)$. 
\end{lemma}
\begin{proof}
Consider the construction of $\mathcal{G}^*_{L,f}$ described in the proof of Lemma \ref{lem:smaller-covering-family}. Let $\mathcal{S}=\mathcal{D}_{L,f}$ and let $\mathcal{H}_{\mathcal{S}}=\{h: [N] \to [2ab\log N]\}$ be the covering graph family for $\mathcal{S}$. 
Fix $(P,F) \in \mathcal{D}_{L,f}$. We claim that the only subgraphs in $\mathcal{G}^*_{L,f}$ that fully avoid a fixed set of exactly $F=\{e_{j_1}, \ldots, e_{j_f}\}$ edge faults is given by the subset of subgraphs $\mathcal{G}_F=\{ G_{h,h(e_{j_1}),\ldots, h(e_{j_f})} ~\mid~ h \in \mathcal{H}_{\mathcal{S}}\}$. 
To see this consider a subgraph $G'=G_{h,i_1,\ldots, i_f}$ where there exists $e_{j_\ell}$ such that $h(e_{j_\ell})\notin \{i_1,\ldots, i_f\}$. In this case,   we have that $e_{j_{\ell}}\in G'$. 
Since $\mathcal{G}_F$ consists of exactly one subgraph per hash function in $\mathcal{H}_{\mathcal{S}}$, we get that $|\mathcal{G}_F|=O(\log(|\mathcal{D}_{L,f}|)=O(\log N)$.
\end{proof}

\end{document}